\documentclass[journal]{IEEEtran}
\usepackage{cite}
\usepackage{enumitem}
\usepackage{amsmath,amssymb,amsfonts}
\usepackage{amsthm}
\usepackage{algorithmic}
\usepackage{graphicx}
\usepackage{textcomp}
\usepackage{xcolor}
\usepackage{subfigure}
\usepackage{diagbox}
\usepackage[utf8]{inputenc}
\usepackage[english]{babel}
\usepackage{booktabs}
\usepackage{multirow}
\usepackage{array}

\newtheorem{claim}{\bf Claim}
\newtheorem{assumption}{\bf Assumption}

\newtheorem{rmk}{\bf Remark}
\newtheorem{lemma}{\bf Lemma}

\newtheorem{setting}{\bf Setting}
\usepackage{pifont}
\usepackage{adjustbox,booktabs,multirow}

\definecolor{orange}{RGB}{255,107,0}

\usepackage[]{color-edits}
\addauthor{sun}{blue}
\addauthor{song}{blue}

\def\bTheta{\boldsymbol \Theta}

\def\bp{\mathbf p}
\def\bh{\mathbf h}

\def\bP{\mathbf P}

\begin{document}
	\title{To Supervise or Not: How to Effectively Learn  Wireless Interference Management Models?}

\author{Bingqing Song, Haoran Sun, Wenqiang Pu, Sijia Liu, and Mingyi Hong \thanks{B. Song, H. Sun and M. Hong are with Department of Electrical and Computer Engineering, University of Minnesota, Minneapolis, MN 55455, USA. W. Pu is with Shenzhen Research Institute of Big Data, Shenzhen, China.  S. Liu is with the CSE Department, Michigan State University, East Lancing, MI. \newline {A short version of this paper \cite{BS2021} has been submitted to SPAWC 2021.}}
}

	\maketitle
	
	\begin{abstract}
Machine learning  has become  successful in solving wireless  interference management problems. Different kinds of deep neural networks (DNNs) have been trained to accomplish key tasks such as power control, beamforming and admission control. There are two popular training paradigms for such DNNs-based interference management models: supervised learning (i.e., fitting labels generated by an optimization algorithm) and unsupervised learning (i.e., directly optimizing some system performance measure). Although both of these paradigms have been extensively applied in practice, due to the lack of any theoretical understanding about these methods, it is not clear how to systematically understand and compare their performance. 


In this work, we conduct theoretical studies to provide some in-depth understanding about these two training paradigms. 
First, we show a somewhat surprising result, that for some special power control problem, the unsupervised learning can perform much worse than its  supervised counterpart, because it is more likely to stuck at some low-quality local solutions. We then provide a series of theoretical results to further understand the properties of the two approaches. Generally speaking, we show that when high-quality labels are available, then the supervised learning is less likely to be stuck at a solution than its unsupervised counterpart. Additionally, we develop a semi-supervised learning approach which properly integrates these two training paradigms, and can effectively utilize limited number of labels to find high-quality solutions. 
To our knowledge, these are the first set of theoretical results trying to understand different training approaches in learning-based wireless communication system design.  
	\end{abstract}
	
	
	\section{Introduction}
	\label{sec:intro}
	\noindent{\bf Motivation.} In recent years, machine learning techniques have become increasingly successful in solving a variety of wireless  interference management problems. Different kinds of deep neural network (DNN), such as fully connected network  (FCN) \cite{sun2018learning}, recurrent neural network (RNN) \cite{8242643}, graph neural network (GNN) \cite{shen2019graph,eisen2020optimal}
	have been designed to accomplish key tasks such as power control \cite{liang2019towards}, beamforming \cite{sun2018learning}, 
	MIMO detection \cite{8227772}, among others. These DNN based models are capable of achieving competitive and sometimes even superior performance compared to the state-of-the-art optimization based algorithms \cite{liang2019towards}. 
	
	
	However, despite its success, there is still a lack of basic understanding about {\it why} DNN based approaches work so well for this class of wireless communication problems -- after all, the majority of interference management problems (e.g., beamforming) are arguably more complex than a typical machine learning problem such as image classification. It is widely believed that, exploiting task-specific properties in designing network architectures and training objectives can help reduce the network complexity and input feature dimension \cite{liang2019towards}, boost the training efficiency  \cite{liang2019towards}, and improve the expressiveness \cite{sun2018learning}. 
	
	The overarching goal of this research is to understand how problem-specific properties affect the DNN training outcomes. More concretely, we attempt to provide an in-depth analysis about two state-of-the-art training paradigms -- the supervised learning and the unsupervised learning approaches. Our theoretical analysis demonstrates that, key metrics in a wireless system, such as interference levels, indeed play some important roles  in characterizing the solution quality of different training approaches. 
	
	For the simplicity of presentation, throughout the paper, we will utilize the classical weighted sum rate (WSR) maximization problem in single-input single output (SISO) interference channel as a working example, but we believe that our approaches and the phenomenon we observed can be extended to many other related problems as well.

	\noindent{\bf Problem Statement and Contributions.} Consider training DNNs for power control, or more generally for beamforming. There are two state-of-the-art approaches for training: \\
	\noindent 1) {\it supervised learning (SL)}, in which labels of high-quality power allocations are generated by an optimization algorithm, then the training step minimizes the mean square error (MSE) between the DNN outputs and the labels \cite{sun2018learning}; \\
	\noindent 2)  {\it unsupervised learning (UL)}, which directly optimizes some system utilities such as the WSR \cite{liang2019towards}, over the training set. 
	
	It is worth highlighting that the above mentioned UL approach has been designed specifically for the interference management problem, because the task of wireless system utility optimization problem (which includes the WSR maximization as a special case) offers a natural training objective to work with. Since it does not require any existing algorithms to help generate high-quality labels, it is much preferred when training samples are difficult to generate. {On the other hand, the associzated training objective appears to be difficult to optimize, since the WSR is a highly non-linear function with respect to (w.r.t.) the transmit power or the beamformer, which in turn is a highly non-linear function of the DNN parameters.}

	Which training method shall we use in practice? Can we rigorously characterize the behavior of these methods? Is it possible to properly integrate these two approaches to yield a more efficient training procedure? Towards addressing these questions, this work makes the following key contributions:\\
	\noindent \ding{182} We focus on the SISO power control problem in interference channel (IC), and identify a simple 2-user setting, in which UL approach  has {\it non-zero probability} of getting stuck at low-quality solutions (i.e., the local minima), while SL approach always finds the global optimal solution;
	
	\noindent \ding{183}  We provide analysis to understand properties of UL and SL for DNN-based SISO-IC problem, by partly leveraging recent advances in training overparameterized deep networks. Roughly speaking, our results indicate that when high-quality labels are provided, SL should outperform UL in terms of solution quality. 
	
	
	\noindent \ding{184} In an effort to leverage the advantage of both approaches, we develop a {\it semi-supervised} training objective, which regularizes the unsupervised objective by using a few labeled data points. Surprisingly, by only using a small fraction ($< 2\%$) of samples of SL approach, the semi-supervised approach is able to avoid bad local solutions and attain similar performance as supervised learning.

	To the best of our knowledge, this work provides the first in-depth understanding about the two popular  approaches for training DNNs for wireless communication problem.

	\noindent{\bf Notations.}  We follow the conventional notation in signal processing, $x$, $\mathbf{x}$, and $\mathbf{X}$ denote scalar, vector, and matrix, respectively. For a given matrix $\mathbf{X}$, we use $\|\mathbf{X}\|_{F}$ to denote the Frobenius norm;
	$\sigma_{\min }(\mathbf{X})$ and $\sigma_{\max }(\mathbf{X})$ denote its smallest and largest singular value, respectively. We use $\|\cdot\|_2$ to denote the $l_2$ norm of a vector, $\otimes$ to denote the Kronecker product. The notation $[N]$ denotes the set $\{1,\cdots, N\}$ and $\mathbf{I}_n\in \mathbb{R}^{n\times n}$ denotes the identity matrix.

	\section{Preliminaries}
	\label{sec:preliminaries}
	Consider a wireless network consisting of $K$ pairs of transmitters and receivers. Suppose each transmitter (and receiver) equips with a single antenna, denote $h_{kj}\in\mathbb{C}$ as the channel between the $k$th transmitter and the $j$th receiver, $p_{k}$ as the power allocated to the $k$th transmitter, $P_{\max}$ as the budget of transmitted power, and $\sigma^{2}$ as the variance of zero-mean Gaussian noise in the background. Further, we use $w_{k}$ to represent the prior importance of the $k$th receiver, then the classical WSR maximization problem can be formulated as: 
	\begin{align}\label{eq:wsr}
		\max _{p_{1}, \ldots, p_{K}} & \sum_{k=1}^{K} w_{k} \log \left(1+\frac{\left|h_{kk}\right|^{2} p_{k}}{\sum_{j \neq k}\left|h_{k j}\right|^{2} p_{j}+\sigma_{k}^{2}}\right) : = R(\mathbf{p}; |\mathbf{h}|)\nonumber\\
		\text { s.t. } \ & 0 \leq p_{k} \leq P_{\max }, \forall k=1,2, \ldots, K,
	\end{align}
	where $\mathbf{h}:=\{h_{kj}\}$ collects all the channels; $|\cdot|$ is the componentwise absolute value operation;  
	and $\mathbf{p}:=(p_1,p_2,\ldots,p_K)$ denotes the transmitted power of $K$ transmitters. Problem~\eqref{eq:wsr} is a fundamental problem in wireless communication, and is known to be NP-hard \cite{luo2008dynamic} in general. For problem \eqref{eq:wsr} and its generalizations such as the beamforming problems in MIMO channels, many iterative optimization based algorithms have been proposed, such as waterfilling algorithm~\cite{4760244}, interference pricing~\cite{5230846} , WMMSE~\cite{shi2011iteratively}, and SCALE~\cite{5165179}.
	
	Recently, there has been a surge of works that apply DNN based approach to find good solutions for problem~\eqref{eq:wsr} as well as its extensions~\cite{sun2018learning,liang2019towards}. 
	{{Although these works differ from their problem settings and/or DNN  architectures, they 
			all use either the SL, UL, or some combination of the two to train the respective networks.  
			Below let us take problem \eqref{eq:wsr} as an example and  briefly compare the SL and UL approaches.}}

	\noindent
	$\bullet$ \textbf{Training Samples:} {{
			Both approaches require a collection of the channel information over $N$ different {\it snapshots} as training samples, denoted as $\mathbf{h}^{(n)}\in\mathbb{R}^{K^2},n\in\mathcal{N}:=[N]$.  However, SL requires an additional $N$ labels $\bar{\mathbf{P} } := \{\bar{\mathbf{p}}^{(n)}\}_{n=1}^N$ corresponding to the $N$ snapshots, which are usually obtained by solving $N$ independent instances of \eqref{eq:wsr} using a proper optimization algorithm, e.g., WMMSE~\cite{shi2011iteratively}. 
			Notice that the quality of such labels may depend on the property of the selected optimization algorithm.
			
	}}
	
	\noindent
	$\bullet$ \textbf{DNN Structure:} { Both approaches can utilize the same DNN structure. Consider an $L$-layer fully connected neural network which takes channel information $|\mathbf{h}^{(n)}|$ as input and outputs the corresponding power allocation $\mathbf{p}^{(n)}\in\mathbb{R}^{K}$. Denote the weights of each layer as  $\left\{\mathbf{W}_{l}\right\}_{l=1}^{L}$, then the output of the $l$-th layer across all samples is denoted as $\mathbf{F}_{l} \in \mathbb{R}^{N \times n_l}$ ($n_\ell$ is the output dimension), given as: 
	\begin{equation}\label{eq:nn}
		\mathbf{F}_{l}:=\left\{\begin{array}{ll}
			\mathbf{H}, & l=0\ \textrm{(input layer)} \\
			a\left(\mathbf{F}_{l-1} \mathbf{W}_{l}\right) & l \in[1:L-1] \ \textrm{(hidden layer)}\\
			\mathbf{P}, & l=L\ \textrm{(output layer)},
		\end{array}\right.
	\end{equation}%
	where $\mathbf{H} \in \mathbb{R}^{N \times K^2}$ stacks all training samples $\{|\mathbf{h}^{(n)}|\}_{n=1}^{N}$, $a(\cdot)$ is a proper component-wise activation function, and 
	\begin{equation}\label{eq:FLa}
	    \mathbf{P}=a(\mathbf{F}_{L-1} \mathbf{W}_{L})\in\mathbb{R}^{N\times K}
	\end{equation}
	denotes the power allocation associated with the corresponding input $\mathbf{H}$. We further define the vectorized power allocation and label $\bar{\bP}$ as
	\begin{equation}
	    \mathbf{f}_L:=\operatorname{vec}(\bP)\in \mathbb{R}^{NK},\; \mathbf{y}:= \operatorname{vec}(\bar{\bP}) \in \mathbb{R}^{NK}.
	\end{equation}
	Compactly, denote $\bTheta:=\{\mathbf{W}_{l}\}_{l=1}^{L}$ and $\mathbf{p}^{(n)}$ as the $n$th row of $\mathbf{P}$, the output of DNN associated with sample $|\mathbf{h}^{(n)}|$ can be expressed as 
    \begin{align}\label{eq:nn_onesample}
        \mathbf{p}^{(n)}& =\mathbf{p}\left(\mathbf{\Theta} ;\left|\mathbf{h}^{(n)}\right|\right)\in \mathbb{R}^K
    \end{align}%
    where $\mathbf{p}\left(\mathbf{\Theta} ;\left|\mathbf{h}^{(n)}\right|\right)$ is the nonlinear mapping defined by~\eqref{eq:nn} with parameter $\mathbf{\Theta}$ and input $|\mathbf{h}^{(n)}|$.
	}

	\noindent
	$\bullet$ \textbf{Training Problems:}  { Both SL and UL utilize the DNN structure defined by~\eqref{eq:nn_onesample}, but the loss functions are different.} In particular, the SL approach
	minimizes the MSE between the  neural network output and the labels:
	\begin{equation}\label{eq:supervised}
	\begin{aligned}
		\min_{\mathbf{\Theta}}&  \quad \sum_{n=1}^{N}\|\mathbf{p}(\mathbf{\Theta}; |\mathbf{h}^{(n)}|)  - \bar{\mathbf{p}}^{(n)}\|^2:=f_{\rm sl}(\mathbf{\Theta})\\
		{\rm s.t.} & \quad {\mathbf 0}\le \mathbf{p}(\mathbf{\Theta}; |\mathbf{h}^{(n)}|) \le \mathbf{P}_{\max}, \; \forall~n\in\mathcal{N}. 
	\end{aligned}
	\end{equation}
	For the UL approach, it directly optimizes the sum of WSRs across all the samples:
	\begin{equation}\label{eq:un-supervised}
	\begin{aligned}
		\min_{\mathbf \Theta} & \quad  \sum_{n=1}^{N}  - R\left(\mathbf p(\mathbf \Theta; |\mathbf h^{(n)}|), |\mathbf h^{(n)}|\right):=f_{\rm ul}(\mathbf \Theta)\\
		{\rm s.t.} & \quad {\mathbf 0}\le \mathbf p(\mathbf \Theta; |\mathbf h^{(n)}|) \le \mathbf{P}_{\max}, \  \forall~n\in\mathcal{N}. 
	\end{aligned}
	\end{equation}
	
	
{ Now we have completed our description of the SL and UL paradigms, let us use a simple example to illustrate their potential performance differences. Fig. \ref{fig1} shows that for a 2-user scenario, the two learning approaches with the same DNN structure (different learning approaches would lead to different $\mathbf{\Theta}$) can have significantly different performance in strong interference case (the precise settings are presented in Sec. \ref{sec:simul}). Then a natural question arises: what are there fundamental differences between the SL training procedure and the UL training procedure? }
\begin{figure}[h]
\centering
\includegraphics[width=0.8\linewidth]{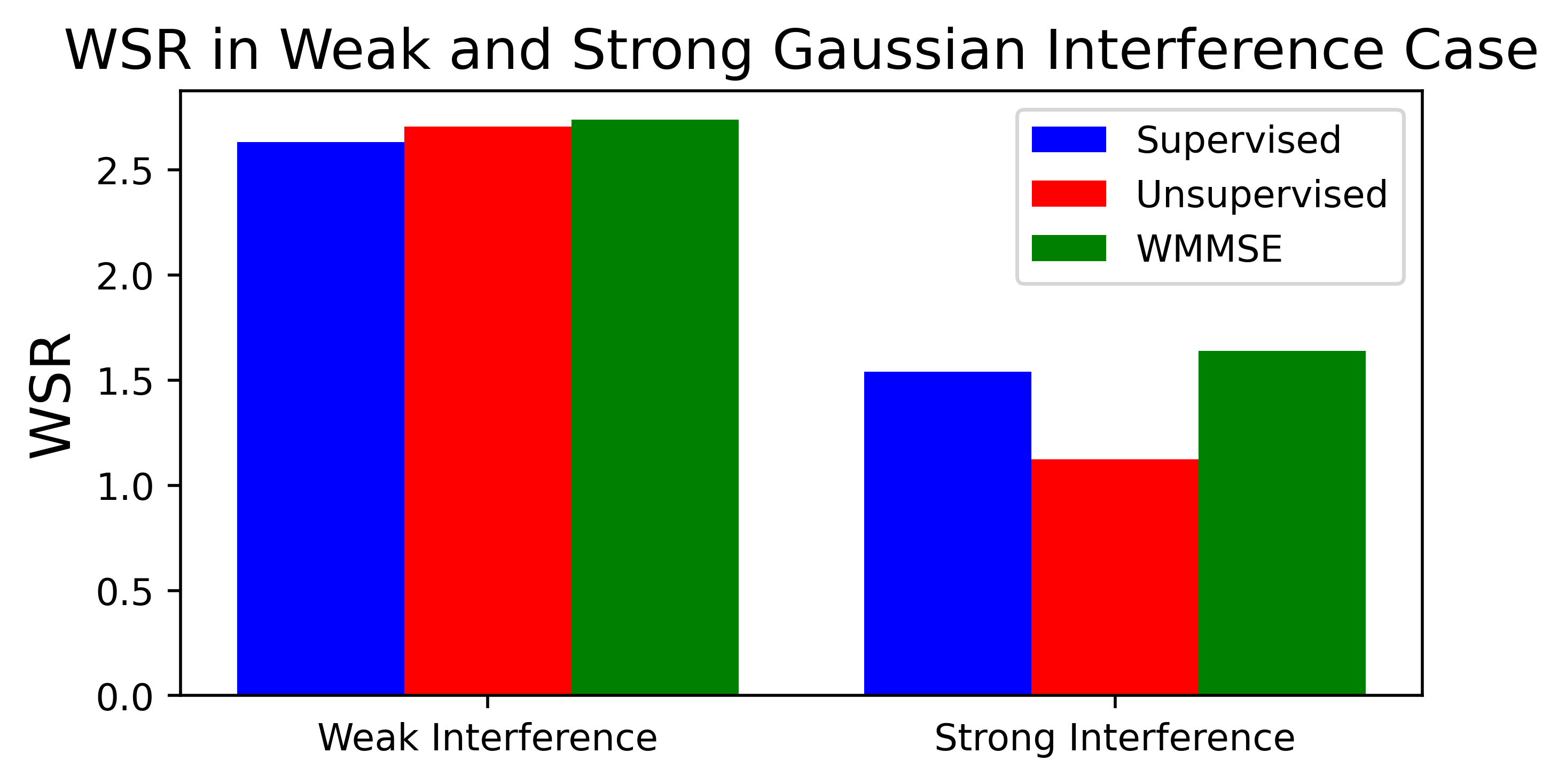}
\caption{\small Comparison between the SL, UL, and WMMSE algorithms, where weak and strong interference cases refer the power of the interference channel being equal and $10$ times of the direct channel, respectively.
}
\label{fig1}
\end{figure}

	\section{The Solution Qualities of SL and UL Problems}
	\label{sec:study}

	

This section investigates the SL and UL paradigms by analyzing properties of their corresponding solution sets. 

{ To begin with, let us compare problems~\eqref{eq:supervised} and~\eqref{eq:un-supervised} in an intuitive way. Problem~\eqref{eq:supervised} can be easily understood as trying to learn the mapping $\bp(\mathbf \Theta; |\mathbf h^{(n)}|)$ under the supervision of labels $\{ \bar{\mathbf{p}}^{(n)}\}$. However, for problem~\eqref{eq:un-supervised}, it directly leverages the WSR maximization formulation~\eqref{eq:wsr} and does not require any labels. Intuitively, this problem is harder to optimize compared with problem~\eqref{eq:supervised} because $R\left(\mathbf p(\mathbf \Theta; |\mathbf h^{(n)}|), |\mathbf h^{(n)}|\right)$ is a {\it composition} of two non-trivial nonlinear functions, i.e., $R(\cdot; |\mathbf h|)$ and  $\mathbf p(\cdot; |\mathbf h|)$. Purely optimizing $R(\cdot; |\mathbf h|)$ (i.e., solving problem~\eqref{eq:wsr}) is known to be NP-hard in general~\cite{luo2008dynamic} and now a nonlinear function $\mathbf p(\cdot; |\mathbf h|)$ is composited within it. Further, problem~\eqref{eq:un-supervised} couples $N$ difficult problems since it tries to find {\it a single} parameter $\mathbf \Theta$ that maximizes the sum of WSR across $N$ snapshots (or say problem instances). Even if each instance of problem is easy to optimize, it may still be difficult to optimize all the coupled problem simultaneously.

}

To make the above intuition precise, { we begin our analysis from a relatively simple 2-user scenario. In such a scenario, the sum rate problem \eqref{eq:wsr} is easy to solve and the solution will be {\it binary}~\cite{Gjendemsjo06,Charafeddine09}. It follows that the {\it optimal} labels for the SL approach can be obtained easily. For simplicity, consider $w_k=\frac{1}{N},\forall~k$, $P_{\max}=1$, $\sigma=1$. Then, under a toy setting (i.e., Setting~\ref{setting:toy} below), the solution quality of the SL and UL approaches is given in Claim~\ref{claim:2x2}. Note that the detailed proof of this claim is give in Appendix \ref{app:claim:2x2}.}

\begin{setting}[Toy Setting]\label{setting:toy}
Consider WSR training with
\begin{itemize}[fullwidth,itemindent=0em,label=$\bullet$]
\setlength{\parskip}{0pt}
    \item \noindent $K=2$, $N=2$, optimal labels $\{\bar{\mathbf{p}}^{(n)}\}_{n=1}^N$;
    \item \noindent one-layer linear neural network, i.e., $\mathbf p^{(n)} = \mathbf \Theta |\mathbf h^{(n)}|$.
\end{itemize}
\end{setting}

	
\begin{figure}
\centering
\includegraphics[height=6cm]{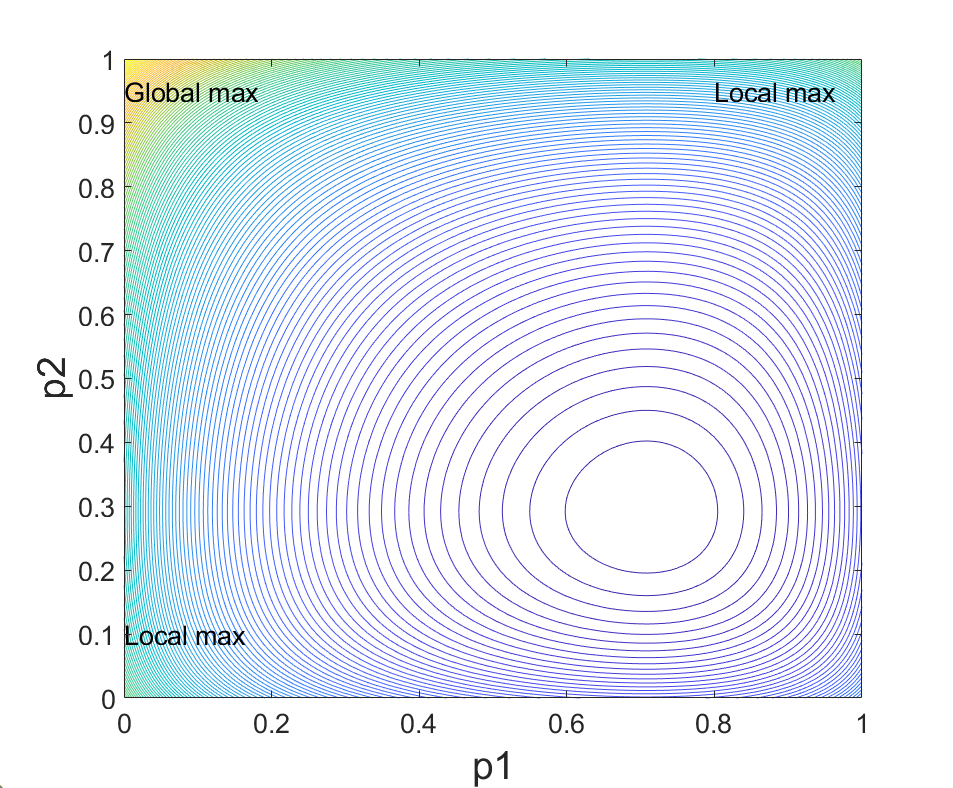}
\caption{\small For two-user IC with 2 snapshots, the true labels $\bar{\mathbf{p}}^{(1)}=(0,1)$, $\bar{\mathbf{p}}^{(2)}=(1,0)$. For both users, keep the sum of label of each snapshot to be $1$ (since we know that the global optimal solution has this structure), that is $\mathbf{p}^{(1)}=(p_1,1-p_1)$, $\mathbf{p}^{(2)}=(p_2,1-p_2)$. We plot the sum-rate of the two snapshots. The upper right and lower left corners are local maxima while the upper left is the global maximum. }
\label{fig2}
\end{figure}
	
	
\begin{claim}\label{claim:2x2} 
{ 	Consider problems~\eqref{eq:supervised} and ~\eqref{eq:un-supervised} under Setting~\ref{setting:toy}. Then, there exist $\{ \mathbf h^{(n)}\}_{n=1}^N$ that 
\begin{enumerate}[fullwidth,itemindent=0em,label=$\bullet$]
    \item The SL problem~\eqref{eq:supervised} has infinitely many local optimal solutions and all of them are globally optimal, and they achieve zero training loss; 
    \item The UL problem~\eqref{eq:un-supervised} has at least two local optimal solutions and at least one of them is sub-optimal.
\end{enumerate}
}
\end{claim}
{
\noindent{\it Proof Outline.} The detailed proof can be found in Appendix \ref{app:claim:2x2}. Hereafter, we give an outline. By constructing $\{ \mathbf h^{(n)}\}_{n=1}^2$ satisfying proper conditions in \eqref{eq:example} and \eqref{eq:key:construction}, we can verify that the global optimal power allocation for each snapshot is $\bar{\mathbf{p}}^{(1)}=(0,1),\  \bar{\mathbf{p}}^{(2)}=(1,0)$.
By analyzing the landscapes of the SL and UL losses, it is easy to verify that any local optimal solution of the SL problem \eqref{eq:supervised} is globally optimal since $f_{\rm sl}(\bTheta)$ is a convex function w.r.t. $\bTheta$ (due the linearity of the neural network). However, for the UL problem \eqref{eq:un-supervised}, there exist at least two types of local optimal solutions $\mathbf \Theta_{\rm global}$ and $\mathbf \Theta_{\rm local}$, 
which generate two different outputs: 
\begin{subequations}
 	\begin{align}
\mathbf{p}(\mathbf \Theta_{\rm global},|\mathbf{h}^{(1)}|) & =(0,1), \;     \mathbf{p}(\mathbf \Theta_{\rm global},|\mathbf{h}^{(2)}|) =(1,0), \label{eq:opt:sol}\\
	\mathbf{p}(\mathbf \Theta_{\rm local},|\mathbf{h}^{(1)}|) & =  \mathbf{p}(\mathbf \Theta_{\rm local},|\mathbf{h}^{(2)}|) =(1,0). \label{eq:subopt:sol}
\end{align}    
\end{subequations}
Since the SL and UL problems have the same neural netwrok structure, it is straightforward to see $\mathbf \Theta_{\rm global}$ corresponds to optimal solutions of the SL problem \eqref{eq:supervised}. To prove the existence of $\mathbf \Theta_{\rm local}$, we follow the two steps outlined below (for detailed argument, see Appendix \ref{app:claim:2x2}). In the subsequent discussion, sometimes it will be convenient to express the loss functions in terms of the power allocation $\mathbf{P}$ rather than $\mathbf{\Theta}$. In these cases, we will denote the loss function as $\widetilde{f}_{\rm sl}(\mathbf{P})$ and $\widetilde{f}_{\rm ul}(\mathbf{P})$, respectively.

\noindent{\bf Step 1.} Denote $\bP^* := [1, 0; 1, 0]$ as the suboptimal power allocation, we will show that there exists a region $N_{\epsilon}(\mathbf{P}^*):=\{\bP:\|\bP-\bP^{*}\|\leq\epsilon, 0\leq \bP\leq 1 \}$ around $\mathbf{P}^*$ such that 
\begin{align}\label{eq:key:fx}
    \widetilde{f}_{\rm ul}(\bP^*)- \widetilde{f}_{\rm ul}({\bP}) \leq 0,  \forall \bP\in N_{\epsilon}({\bP}^*).
\end{align}

\noindent{\bf Step 2.} Next we show that for every $\widetilde{\bTheta}$ satisfying  $$\mathbf{P}\left(\widetilde{\bTheta},|\mathbf{H}|\right)=\left[\mathbf{p}\left(\widetilde{\bTheta},\left|\mathbf{h}^{(1)}\right|\right);\mathbf{p}\left(\widetilde{\bTheta},\left|\mathbf{h}^{(2)}\right|\right)\right]=\mathbf{P}^{*},$$
there exists a region $N_{\delta}(\widetilde{\bTheta})$ such that for all $\bTheta \in N_{\delta}(\widetilde{\bTheta})$, the corresponding $\mathbf{P}(\bTheta,|\mathbf{H}|)$ falls in $N_{\epsilon}\left(\mathbf{P}^{*}\right)$. That is, $f_{\rm ul}({\bTheta})>f_{\rm ul}(\widetilde{\bTheta}),\forall {\bTheta}\in N_{\delta}(\widetilde{\bTheta})$, and hence $\widetilde{\bTheta}$ is a local minimum.
$\qed$
}

	 

{
Fig.~\ref{fig:land} illustrates the landscapes of the UL loss under Setting~\ref{setting:toy}. A few remarks are listed below.
\begin{rmk}
Claim~\ref{claim:2x2} states that the UL problem \eqref{eq:un-supervised} is more challenging compared with the SL problem \eqref{eq:supervised} (at least for the toy setting 1). This is mainly because the UL approach treats the training as a single problem, which composites the neural network training and WSR maximization together. Even for a simple problem instance decribed by Setting~\ref{setting:toy}, the UL problem \eqref{eq:un-supervised} still has a sub-optimal solution.
\end{rmk}
}


	
	
\begin{rmk}\label{rmk:non-zero}
  In Claim \ref{claim:2x2}, we constructed a specific example and showed the existence of local minima for unsupervised training. Generally speaking, it is possible to show that this phenomenon appears, with non-zero probability (w.r.t. the channel generation process), when the interference channels are strong. However, the proof will be long and technical, and will not reveal too much insight beyond what has already been shown in Claim \ref{claim:2x2}. Therefore, we will not provide such a proof in this paper. 
\end{rmk}

\begin{figure} 
	\centering
	\subfigure[$f =10$]{
		\begin{minipage}[t]{0.48\linewidth}
			\centering
			\includegraphics[width=\linewidth]{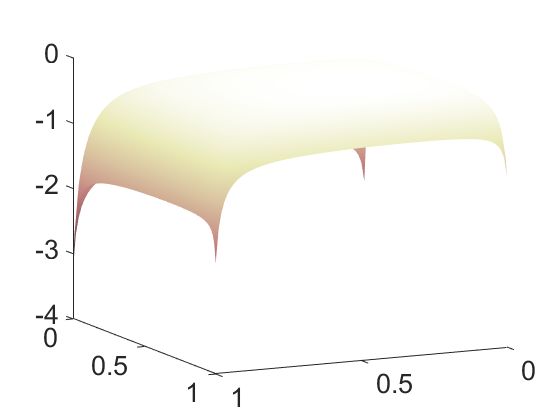}
		\end{minipage}%
	}%
	\subfigure[$f = 0.1$]{
		\begin{minipage}[t]{0.48\linewidth}
			\centering
			\includegraphics[width=\linewidth]{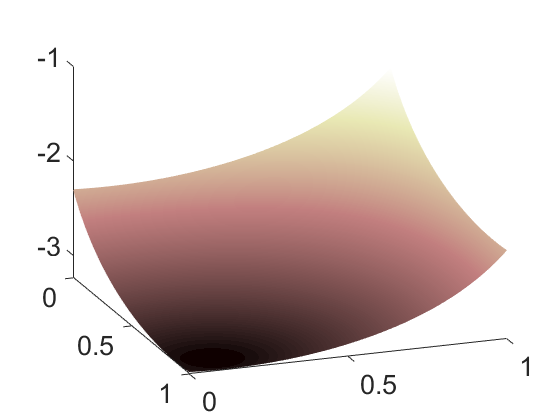}
		\end{minipage}%
	}%
	\caption{\small Landscapes of the UL loss with different interference channel realizations, where  $|h_{11}^{(1)}|=|h_{22}^{(2)}|=1$, $|h_{22}^{(1)}|=|h_{11}^{(2)}|=2$, and $|h_{12}^{(1)}|=|h_{21}^{(1)}|=|h_{12}^{(2)}|=|h_{21}^{(2)}|=f$ ($f=0.1$ or $10$).}
	\label{fig:land}
	\vspace{-0.5cm}
\end{figure}

{
Next, we study the solution quality of the SL and UL problems in a more general case, where the system consists of  arbitrary number of users, and a nonlinear neural network is used; see Setting~\ref{setting:gen} below.

\begin{setting}[General Setting]\label{setting:gen}
Consider WSR training with
\begin{itemize}[fullwidth,itemindent=0em,label=$\bullet$]
\setlength{\parskip}{0pt}
    \item \noindent $K\geq2$, $N\geq2$;  $\bar{\mathbf{p}}^{(n)},n\in\mathcal{N}$ are stationary labels, defined as solutions that satisfy KKT condition of problem \eqref{eq:wsr};
    \item \noindent $L$-layers ($L\geq 1$) nonlinear neural network.
\end{itemize}
\end{setting}

For $K>2$, the WSR problem~\eqref{eq:wsr} is NP-hard in general, and classical algorithms such as WMMSE \cite{shi11WMMSE_TSP} can find stationary solutions for each individual snapshot. Hence only stationary labels are considered in Setting~\ref{setting:gen}. To reveal the fundamental difference between the SL and UL problems, we impose additional assumptions on the neural network architecture and the capabilities of the training procedure; see Assumption~\ref{assp:gen} below. 
\begin{assumption}\label{assp:gen}
    Consider Setting~\ref{setting:gen} and assume 
\begin{enumerate}[fullwidth,itemindent=0em,label=\roman*)]
\setlength{\parskip}{0pt}
    \item The nonlinear neural network is expressive enough such that there exists $\bTheta^*$ with $f_{\rm sl}(\bTheta^*)=0$;
    \item There exists training algorithm that the SL training can find such a  $\bTheta^*$ as defined in the previous item.
\end{enumerate}
\end{assumption}
Assumption~\ref{assp:gen} is related to the expressivity of the neural network and the  performance for a training algorithm, and they appear to be relatively stringent. However, recent advances in deep learning suggest that these conditions could be satisfied for special neural networks and  algorithms. Specifically, zero training loss has been verified when the neural network is {\it {overparameterized}}; see e.g., \cite{zhang2016understanding}. {Further, it has been shown in \cite{allen2019convergence,du2019gradient} that, for certain special neural networks, the gradient descent (GD) algorithm can indeed find such a global optimal solution. However, for various reasons to be discussed shortly, these works cannot be directly applied to analyze our problems. The condition under which   Assumption~\ref{assp:gen}  can be satisfied will be studied in Section~\ref{sec:opt}. Hereafter, we continue studying the solution quality between the SL and UL problems. Under Assumption~\ref{assp:gen}, it is sufficient to restrict our focus on studying the solution set of SL with zero training loss, denoted as
\begin{equation}\label{soluset:sl}
\mathcal{S}:=\{\bTheta \mid f_{\rm sl}(\bTheta)=0 \}.
\end{equation}
However, we will see that even with Assumption~\ref{assp:gen}, it is still difficult to characterize the UL training problem~\eqref{eq:un-supervised} , so we still need to use the stationary solutions, defined below, when studying this problem:
\begin{equation}\label{soluset:ul}
\mathcal{U}:=\{\bTheta \mid \bTheta\textrm{ is a stationary solution of pronlem~\eqref{eq:un-supervised}} \}.
\end{equation}
The relation between $\mathcal{S}$ and $\mathcal{U}$ is given below. 
	
	
\begin{claim}\label{claim:stationarity}
Suppose Assumption~\ref{assp:gen} hold under Setting~\ref{setting:gen} and consider problems~\eqref{eq:supervised} and~\eqref{eq:un-supervised}. Then, $\mathcal{S}\subseteq\mathcal{U}$.
\end{claim}

\begin{proof}
The proof is based on analyzing the properties of the stationary solutions of problems \eqref{eq:wsr},  \eqref{eq:supervised}, and \eqref{eq:un-supervised}. Please see Appendix~\ref{app:claim:zero:loss} for details.
\end{proof}

\begin{rmk}
Claim~\ref{claim:stationarity} shows that if we impose some additional assumptions on the SL approach, i.e., stationary labels, expressive neural networks, and good training algorithm, then the solution set of the SL approach is no larger than that of the UL approach. So in the ideal case where each label $\bar{\mathbf{p}}^{(n)}$ exactly maximizes the corresponding instance of problem \eqref{eq:wsr}, the SL can find a neural network that {\it simultaneously} optimizes all training instances. On the other hand, it is possible that the UL approach can be trapped at undesirable sub-optimal solution. Note that Claim \ref{claim:stationarity} is a generalization of Claim \ref{claim:2x2}, because Assumption \ref{assp:gen} made in Claim \ref{claim:stationarity} is all trivially satisfied for the SL approach in Claim \ref{claim:2x2}.
\end{rmk}

Based on the above discussion, we see that to fully understand the implication of Claim~\ref{claim:stationarity}, we need to know when Assumption~\ref{assp:gen} will be satisfied. Towards this end, next we will show how one can properly construct a neural network and choose a training algorithm, so that SL training can achieve zero training loss. }
}

	

	\section{The Optimization Performance of SL and UL}\label{sec:opt}
	In this section, we show that it is possible to construct a special 
	neural network and a special training algorithm, { so that Assumption~\ref{assp:gen} can be satisfied.}


\subsection{Neural Network Structure}\label{preliminary}

	
	

	Throughout this section, we will use the notations introduced in Sec. \ref{sec:preliminaries} to represent a neural network. 
{ Note that both problems \eqref{eq:supervised} and  \eqref{eq:un-supervised} have $N$ constraints (one for each sample), which implies the common training algorithm, such as the (stochastic) gradient descent, cannot be directly used. In practice, one often constructs some special structure for the DNN, so that the constraints can be automatically satisfied. A popular and practical approach is to choose a special activation function $b(\cdot)$ in the last layer to enforce feasibility. That is, the last layer output in~\eqref{eq:FLa} is modified to}  
\begin{align}
	 \mathbf{F}_L = b\left(\mathbf{F}_{L-1}\mathbf{W}_{L} \right),
\end{align}
where $b(\cdot)$ is a component-wise activation function (possibly different from $a(\cdot)$), and its output always lies in $[0,P_{\max}]^{N\times K}$. Some choices of $b(\cdot)$ are listed below.
	 	
	 	

	 \noindent $\bullet$  {\bf{(Sigmoid)}} { Sigmoid is a commonly used activation in neural network which is smooth and has bounded output\cite{liang2019towards}:}
	 \begin{align}\label{eq:sigmoid}
	     {\rm Sigmoid}(x) := P_{\max}\cdot \frac{1}{1+e^{-x}}.
	 \end{align}
	 It is easy to see that regardless of the value of $x$, the output of the sigmoid function lies in $[0, \; P_{\max}]$.
	 
	 \noindent $\bullet$ {\bf(Smoothed Clipped ReLU)} The clipped ReLU activation function is given below~\cite{sun2018learning}:
	 \begin{equation}\label{eq:clipped}
	     {\rm ReLU}(x) := \min(\max (0,x),P_{\max}).
	 \end{equation}
	 Although this function is non-smooth, it still ensures that the output lies in $[0, \; P_{\max}]$.
	 
{ 
In this work, we introduce a smoothed version of the clipped ReLU function, named \textit{Smoothed Clipped ReLU (SCReLu)}, expressed as below:
	\begin{equation}\label{eq:soft}
	\mathrm{SCReLu}(x) = \left\{\begin{array}{ll}
			\alpha\cdot (e^{\frac{1}{\alpha}\cdot x}-{1}), \  x < {0},\\
			x   , \  x\in [{0},\;{P_{\max}}],\\
			 {P_{\max}}+\alpha\cdot (1-e^{\frac{P_{\max}-x}{\alpha}}), \  \textrm{o.w.}
	\end{array}\right.
	\end{equation}
Compared with the Sigmoid and the clipped ReLU function, SCReLu function is not only smooth, but has enough curvature in the feasible region $[0, P_{\max}]$. Such a property will be useful in our analysis later. { Note that the output now lies in $[-\alpha, \; P_{\rm max}+\alpha]$, which slightly violates the power constraint.} With such modification for the last layer, both SL and UL problems becomes unconstrained problems. To take into account the modified last layer, and to distinguish the DNN output introduced in~\eqref{eq:nn_onesample}, 
we use notation $\mathbf{q}^{(n)}:=\mathbf{q}(\left(\boldsymbol{\Theta} ;\left|\mathbf{h}^{(n)}\right|\right))$ to denote the DNN output associated with sample $\mathbf{h}^{(n)}$. This way, the unconstrained version of SL and UL training loss $f_{\rm SL}$ and $f_{\rm UL}$ can be respectively expressed as below: 
\begin{subequations}
\begin{align}
f_{\rm SL}(\bTheta) &:= \frac{1}{2}\sum_{n=1}^{N}\left\|\mathbf{q}\left(\boldsymbol{\Theta} ;\left|\mathbf{h}^{(n)}\right|\right)-\bar{\mathbf{p}}^{(n)}\right\|^{2}, \label{eq:vecsup} \\
f_{\rm UL}(\mathbf \Theta)&:=\sum_{n=1}^{N}  - R\left({\mathbf q}(\mathbf \Theta; |\mathbf h^{(n)}|), |\mathbf h^{(n)}|\right)\label{eq:unconstrained-unsupervised}.
\end{align}
\end{subequations}
}%
	 
	 \subsection{The Training Algorithms}\label{sec:alg}	
	 Next, let us discuss the training algorithm. 
	 It is widely accepted that the GD-based algorithms are the state-of-the-art algorithms to optimize neural networks \cite{ruder2016overview}. Therefore, we will work with the classic GD algorithm, and understand how/if it can help us find good solutions for optimizing $f_{\rm SL}(\mathbf{\Theta})$ and $f_{\rm UL}(\mathbf{\Theta})$, as expressed in \eqref{eq:vecsup} -- \eqref{eq:unconstrained-unsupervised} respectively (where the SCReLU activation function \eqref{eq:soft} is used). 
	 Specifically, at the $m$-th iteration of training, we denote  $\bTheta^{m}:=\{ \mathbf{W}_{l}^{m}\}_{l=1}^{L}$ as a collection of the neural network parameters. Then, the GD updates for the SL and UL are given below: 
	 \begin{subequations}
	 \begin{align}
	     \bTheta^{m+1} & = \bTheta^{m} - \eta  \times  \nabla f_{\rm SL}(\bTheta^m), \; m=1,2, \cdots \label{eq:SL:GD}\\
	     \bTheta^{m+1} & = \bTheta^{m} - \eta  \times  \nabla f_{\rm UL}(\bTheta^m), \; m=1,2, \cdots \label{eq:UL:GD}
	 \end{align}    
	 \end{subequations}
    Note that more complicated algorithms such as stochastic gradient descent (SGD) can be considered as well. However the analysis of these variants will be  challenging, and they may not reveal new insights on the difference between SL and UL. Therefore, we will not focus on other variants of GD in this work.

\subsection{Assumptions}\label{sec:assumption}	
	
Let us make the following assumptions about the neural network and  the activation function in the first $L-1$ layers.
\begin{assumption}\label{ass:width}
Assume that the widths of layers satisfy the following condition:
	\begin{align}
	  n_1 \geq N,  \; n_1 \geq n_{2} \geq n_{3} \geq \ldots \geq n_{L}\ge 1.
	\end{align}
	\end{assumption}

	\begin{assumption}\label{ass:activation}
		{\it Assume that the activation function $a(\cdot)$ in \eqref{eq:nn} satisfies the following condition:
		\begin{align*}
		    1)\; a^{\prime}(x) \in [\gamma,1], \;\; 2)\; |a(x)| \leq|x|, \; \forall~x \in \mathbb{R}, \;\; 3)\; a^{\prime} \; \text {is } \beta \text {-Lipschitz}
		\end{align*}
		where $\gamma\in (0,1)$ and $\beta>0$ are some constants.
		}
	\end{assumption}
		
The above two assumptions provide specifications about the training problem. Assumption \ref{ass:width} requires the number of neurons in the first layer to be at least as large as the sample size, and the following layers have decreasing widths. Assumption \ref{ass:activation} can be satisfied by some specifically constructed activation functions, such as the parameterized ReLU function smoothed by a Gaussian kernel~\cite{nguyen2020global}. 

		
{ For simplcity, let us define some useful notations which are related to the singular values of the initial weight matrices. Denote $\underline{\lambda}_{l}:=\sigma_{\min }\left(W_{l}^{0}\right)$, $\underline{\lambda}_{i \rightarrow j} :=\prod_{l=i}^{j} \underline{\lambda}_{l}$. Define:
\begin{equation}\label{eq:lambda}
	\bar{\lambda}_{l}:=\left\{\begin{array}{ll}
		\frac{2}{3}\cdot\left(1+\sigma_{\max }(\mathbf{W}_{l}^{0})\right), & \text { for } l \in\{1,2\}, \\
		\sigma_{\max }(\mathbf{W}_{l}^{0}), & \text { for } l \in\{3, \ldots, L\}.
	\end{array}\right.
\end{equation}
Similarly, denote $\bar{\lambda}_{i \rightarrow j}:=\prod_{l=i}^{j} \bar{\lambda}_{l}$.
Also, two constants that are related to $\mathbf{H}$ are defined below:{\small
\begin{align}\label{eq:alpha:underline}
   \underline{\lambda}_{H}:=\sigma_{\min }\left(a\left(\mathbf{H} \mathbf{W}_{1}^{0}\right)\right),\ \underline{\alpha}_H:={\left({\frac{3}{2}}\right)^L\cdot{\|\mathbf{H}\|_F\prod_{l=1}^{L}\bar{\lambda}_l}}. 
\end{align}}
	
	
The next assumption states how the neural network should be initialized. 
\begin{assumption}\label{ass:init:sl}
Assume the initial weight $\bTheta^{0}$ is generated in such a way that 
\begin{equation}\label{ass:square}
\underline{\lambda}_{H}\geq \max(\Lambda_1,\Lambda_2),
\end{equation}
where $\Lambda_1$ and $\Lambda_2$ are constants given in Appendix~\ref{app:Lambda}, and they are dependent on $\bTheta^0$, $\mathbf{H}$, and $f_{\rm SL}(\bTheta^0)$.
\end{assumption}


Assumption~\ref{ass:init:sl} basically requires that the smallest singular value of the output of the first hidden layer is sufficiently large. Later in Appendix~\ref{app:Lambda}, we will show how such a condition on initialization can be easily satisfied.
}



	

\subsection{Main Results}	
{ 
Our key observation about the differences between the SL and UL training \eqref{eq:SL:GD} and~\eqref{eq:UL:GD} is given below.}
	
	
{ 	
\begin{claim}\label{thm:zero:loss}
Consider WSR training for $f_{\rm SL}(\mathbf{\Theta})$ and $f_{\rm UL}(\mathbf{\Theta})$ under Setting~\ref{setting:gen}, where the training losses are given in~\eqref{eq:vecsup} and~\eqref{eq:unconstrained-unsupervised}. Suppose Assumptions \ref{ass:width}-\ref{ass:init:sl} hold  Then with a small enough stepsize $\eta>0$, we have:
\begin{enumerate}[fullwidth,itemindent=0em,label=$\bullet$]
\item For SL with \eqref{eq:SL:GD}, if the parameter $\alpha$ of SCReLU in~\eqref{eq:soft} satisfies $\alpha\geq\underline{\alpha}_H$, where $\underline{\alpha}_H$ is defined in \eqref{eq:alpha:underline}, there must exist a solution $\bTheta^*:=\lim_{m\rightarrow \infty}\bTheta^m$ such that $f_{\rm SL}\left(\bTheta^{*}\right)=0$. Further, the training loss converges to zero at a geometric rate, 
\begin{equation}
f_{\rm SL}\left(\bTheta^{m}\right) \leq\left(1-\eta \cdot \alpha_{0}\right)^{m} f_{\rm SL}\left(\bTheta^{0}\right)
\end{equation}
where $\alpha_0>0$ is a constant defined in~\eqref{eq:alpha0}.
\item For UL with \eqref{eq:UL:GD}, if $\bTheta^{m}$ is bounded then 
$$\lim_{m  \rightarrow \infty} \nabla f_{\rm UL}\left(\bTheta^{m}\right)=\mathbf{0}.$$
\end{enumerate}
\end{claim}	
}

The detailed proof can be found in Appendix~\ref{app:thm:zero:loss}. The proof relies on the technique developed in~\cite[Theorem 3.2]{nguyen2020global}. The key difference is that, our considered neural network has an different activation function $b(\cdot)$ at the output, and such an activation function does not satisfy Assumption \ref{ass:activation}. Therefore we need to carefully analyze the convergence result with SCReLu activation function.
{
\begin{rmk}
    The key message from Claim~\ref{thm:zero:loss} is that, if GD is used for training, then SL enjoys nice convergence property while UL does not. Specifically, for the UL training, even assuming a strong condition that all weights are bounded during UL training, only convergence to stationary points can be ensured. This is not surprising, considering that the sum-rate problem itself is an NP-Hard problem, even without compositing it with a neural network. 
\end{rmk}
}

\subsection{Discussion of Claim \ref{thm:zero:loss}}\label{sub:discussion}
Some discussion on Claim \ref{thm:zero:loss} are provided below.
	
\noindent\textbf{Comparison of UL and SL: }{ First of all, the claimed geometric convergence rate for SL training requires the parameter $\alpha$ in SCReLU function satisfying $\alpha\geq\underline{\alpha}_H$, while the UL training does not have such a requirement. Choosing large $\alpha$ does not affect the feasibility in SL training since the neural network outputs can converge to labels $\{\bar{\bp}\}$. For UL training, we only need $\alpha>0$ and the convergence can still be guaranteed. This implies that using SCReLU will only slightly violate the power constraint for UL training.

Second, our analysis shows that there is some significant gap between the theoretical guarantees that can be obtained between SL and UL training. For the UL training, even under bounded sequence assumption, only  stationary solutions can be obtained, while through the SL training one can find a model achieving zero training loss.}


\noindent\textbf{Other Convergence Analysis: } There has been a line of works that discuss the the global convergence of supervised learning. These works use quadratic loss functions and require the scalar output \cite{du2019gradient,oymak2020toward,zou2019improved}. Multiple dimension output is needed in our setting, so the network structures in this line of works are not applicable. The network construction and the basic analysis techniques discussed in this section originate from recent advances in neural network optimization 
	\cite{nguyen2018optimization,nguyen2020global,nguyen2021proof,nguyen2021tight}, which allow output to be a vector. Further, different from previous works, sufficient decrease at each iteration is ensured by Assumption \ref{ass:width}, because the smallest singular value of gradient is strictly positive in each layer \cite{nguyen2018optimization}. 
		
\noindent\textbf{Extension to Other Wireless Problems:} Besides the WSR maximization, the network structure in Sec.~\ref{preliminary} allows multi-dimensional output, which make it cloud be used in other wireless communication problems, like MIMO-detection~\cite{samuel2017deep} and signal estimation~\cite{ye2017power}.

 \section{A semi-supervised learning approach}
 \label{sec:semi}
 { From the previous section, we know that under a few regularity assumptions, the SL enjoys fast convergence to zero loss. If high-quality labels are available, potentially SL would perform much better than UL. 
 However the drawback of SL is that, finding high-quality labels can be computationally costly. Is there a way to design a proper learning strategy that only requires a few labels, while still achieving the state-of-the-art performance?  
 In this section, we address this by proposing a {\it semi-supervised} learning (SSL) strategy.
 }

{ 
As indicated by Claim~\ref{claim:2x2}, UL may get stuck at some local solutions once parameters enter some ``bad'' regions. To alleviate such an issue, we propose to add a label-dependent regularization in the training objective to change the landscape of the loss function. 
Specifically, denote $\mathcal{N}:=[N]$ as index set of channel snapshots $\{|\mathbf{h}^{(n)}|\}$ and let $\mathcal{M}\subseteq \mathcal{N}$ be index set for labeled samples $\{|\mathbf{h}^{(m)}|,\bar{\mathbf{p}}^{(m)}\}$. Then, we construct the following training loss by directly combining \eqref{eq:supervised} and \eqref{eq:un-supervised}:
\begin{equation}\label{eq:semi}
\begin{aligned}
    \min_{\bTheta}& \sum_{n\in\mathcal{N}}f_{\rm ul}^{(n)}(\mathbf{\Theta})
    + \lambda\sum_{m\in\mathcal{M}} f_{\rm sl}^{(m)}(\mathbf{\Theta}),\\ 
    \text {s.t.} & \quad {\mathbf 0}\le \mathbf p(\mathbf \Theta; |\mathbf h^{(n)}|) \le \mathbf{P}_{\max}, \  \forall n\in \mathcal{N} ,
\end{aligned}
\end{equation}
where $\lambda>0$ is a constant which controls the trade-off between two different loss terms, and we have defined:
\begin{align*}
\small
f_{\rm sl}^{(m)}(\mathbf{\Theta})&:=\left\| \mathbf{p}(\bTheta;|\mathbf{h}^{(m)}|)- \bar{\mathbf{p}}^{(m)}\right\|^{2},\\
f_{\rm ul}^{(n)}(\mathbf{\Theta})&:=-R\left(\mathbf{p}(\bTheta;|\mathbf{h}^{(n)}|),|\mathbf{h}^{(n)}|\right),
\end{align*}
as the SL and Ul losses associated with samples $(\mathbf{h}^{(m)},\bar{\mathbf{p}}^{(m)})$ and $\mathbf{h}^{(n)}$ respectively. The SL loss $f_{\rm sl}^{(m)}(\mathbf{\Theta})$ serves as a regularization term to supervise the model, so that the predicted labels cannot be too different from the labels. Intuitively, this regularization term can avoid some local solutions of the UL approach. To make this intuition precise, in the following, we first denote the solution set for SSL problem~\eqref{eq:semi} as 
\begin{equation}\label{soluset:ssl}
\begin{aligned}
    \mathcal{L}=\{\bTheta \mid \bTheta&\textrm{ is a stationary solu. of~\eqref{eq:semi} }\\
    &\textrm{ and }f_{\rm sl}^{(m)}(\mathbf{\Theta})=0,\forall m\in\mathcal{M}\}.
\end{aligned}
\end{equation}
Then, the relationship among solution sets of SL problem~\eqref{eq:supervised}, UL problem~\eqref{eq:un-supervised}, and SSL problem~\eqref{eq:semi} is given in Claim~\ref{claim:ssl}.
}

{ 
\begin{claim}\label{claim:ssl}
Consider problems~\eqref{eq:supervised},~\eqref{eq:un-supervised}, and~\eqref{eq:semi} under Setting~\ref{setting:gen} and denote $\mathcal{M}\subseteq \mathcal{N}$ as the index set for labeled samples for SSL problem~\eqref{eq:semi}. Let $\mathcal{S}$, $\mathcal{U}$, and $\mathcal{L}$ be solution sets defined in \eqref{soluset:sl},  \eqref{soluset:ul}, and~\eqref{soluset:ssl}, respectively. Then, $\mathcal{S}\subseteq\mathcal{L}\subseteq\mathcal{U}$.
\end{claim}
}

\begin{rmk}
    The above claim shows that when labels of part of the training data are available,  the set of stationary solutions of the SSL \eqref{eq:semi} lies between the set of stationary solutions of SL and UL. 
    However, at this point it is not clear how to exactly ensure the zero loss on the regularization term in SSL. A practical way is increasing the penalty parameter $\lambda$ in \eqref{eq:semi} to enforce better fit on the quadratic regularization term.
\end{rmk}
	


 	\section{Simulation Results}
 	\label{sec:simul}
{ 
In this section, we provide preliminary numerical results to illustrate the intuition gained from our theoretical studies. Since the main contribution of this work is about understanding the theoretical relations between the SL and UL approaches, we do not try to be exhaustive in our numerical experiments.}

\noindent{\bf \underline{Data Generation}:} We consider the Rayleigh fading channel model~\cite{sklar1997rayleigh}, and set the number of users $K=5,10,20$. Direct channels $h_{kk}$ and interfering channels $h_{kj},k\neq j$ are generated from zero-mean complex Gaussian distribution $\mathcal{CN}(0,\sigma^2)$, where $\sigma$ denotes the standard deviation. To evaluate the stability of different learning approaches, two representative scenarios are considered. In the first scenario (referred as ``weak interference"), both direct and interfering channels are { generated from $\mathcal{CN}(0,\sigma^2)$ with the same $\sigma=1$.} For the second scenario (referred as ``strong interference"), { direct channels are generated with $\sigma=1$, while the interfering channels have larger standard deviation with $\sigma = 10$.} { Note that Fig.~\ref{fig2} in Sec.~\ref{sec:study} depicts the results with $10$ users.}

{ Also, to evaluate the impact of label quality, we consider the following approach to generate low- and high-quality labels. For high-quality labels, instead of directly using WMMSE, we first pass a given sample $\mathbf{h}^{(n)}$ through a pre-trained GNN model { (by the training method in \cite{shen2019graph}),} and then fine tune the label using WMMSE.} As for low-quality labels, we directly use WMMSE to obtain the labels. 




 

\noindent{\bf \underline{Benchmarks}:}
{
We compare the proposed SSL formulation~\eqref{eq:semi} ($\lambda=1$), referred as {\it regularized SSL}, with three benchmark approaches. The first one is the well-studied WMMSE algorithm~\cite{shi2011iteratively}, and the second one is the UL formulation in~\eqref{eq:un-supervised}. Another SSL approach, referred as \textit{pre-trained SSL}, is also included as the third benchmark, which optimizes the UL formulation~\eqref{eq:un-supervised} with a \textit{pre-trained} initialization, i.e., training over labeled samples to obtain the DNN's initialization. 
}

\noindent{\bf \underline{Neural Network Structure}:}\label{sec:5.2}
A fully connected neural network with $3$ hidden layers is used. The number of neurons in each hidden layers is $200,80,80$, respectively, for both $5$- and $10$-user cases, and $600,200,200$ respectivly, for the $20$-user case. For each hidden layer, { the \textit{clipped ReLU} activation defined in~\eqref{eq:clipped} is used, and the \textit{sigmoid} activation in~\eqref{eq:sigmoid} is used for the output layer.} To stabilize the training process, the \textit{Batch Normalization}~\cite{ioffe2015batch} is used after each hidden layer.
	
\noindent{\bf \underline{Training Procedure}:}	All the three DNN-based approaches use the same neural network structure as specified before. 
Unlabeled samples are used in UL approach, which together with the labeled samples are used in the two SSL approaches. The RMSprop algorithm~\cite{tieleman2012rmsprop} is used as the optimizer, where each mini-batch consists of $200$ unlabeled samples and all the available labeled samples. To evaluate the performance, $1,000$ additional unlabeled samples are generated and their averaged sum rate is used as the performance metric.


\noindent{\bf \underline{Results and Analysis}:}
{
The sum rate of the UL and the two SSL approaches in the strong interference scenario is shown in Fig.~\ref{fig3}, where the legends `Low' and `High'  indicate the quality of labels. The total number of unlabeled and labeled samples are $50,000$ and $400$ for the $10$-user case, and $10,000$ and $100$ for the $5$-user case, respectively. Compared with the UL approach, the regularized SSL significantly improves the sum rate in the $10$-user case, especially with high-quality labels. However, the pre-trained SSL does not bring significant improvement. One possible reason is that only hundreds of labeled samples are not enough to pre-train a {\it good} initialization.}
\begin{figure}[htbp]
		\centering
		\subfigure[Strong Interf. with $K=5$.]{
				\centering
				\includegraphics[width=1.6in]{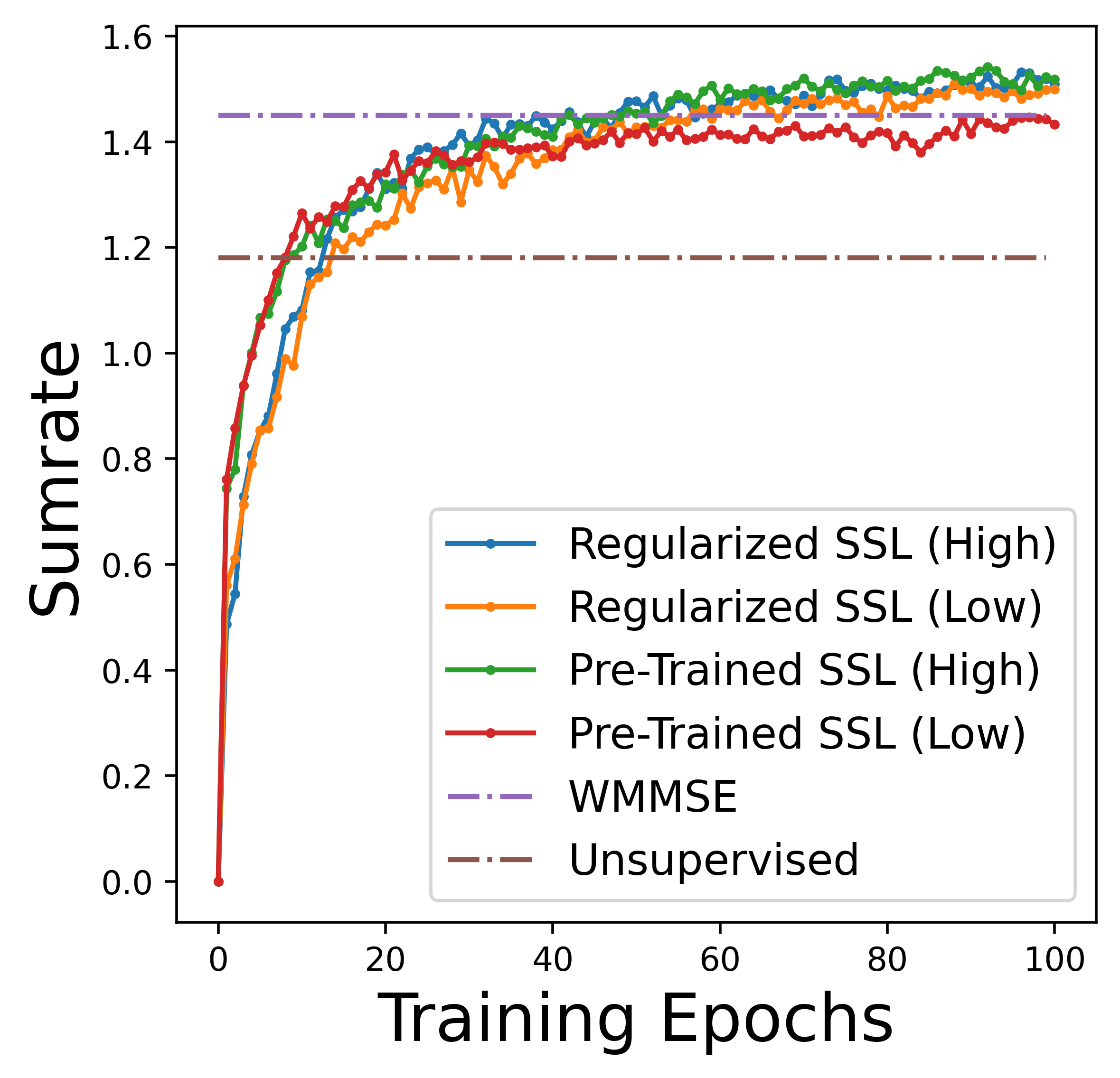}
		}%
		\subfigure[Strong Interf. with $K=10$. ]{
				\centering
				\includegraphics[width=1.6in]{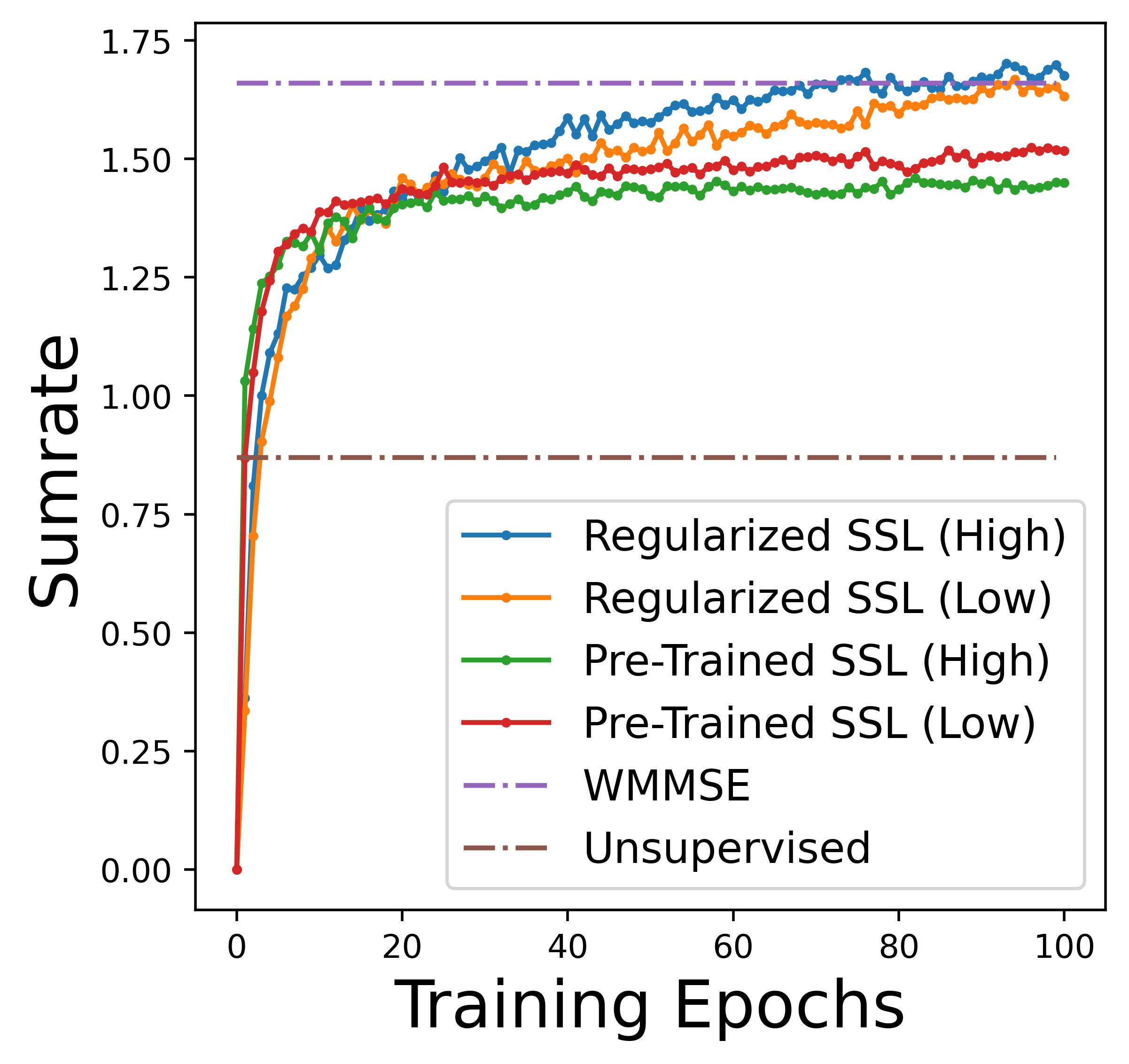}
		}%
		\caption{Sum rate of different approaches in strong interf. scenario.}
		\label{fig3}
\end{figure}

{In the weak interference scenario, we observe that the UL approach can already work well with only a few samples. Hence we set the total number of unlabeled and labeled samples to be $20,000$ and $100$, respectively. The sum rate in this scenario is compared in Table~\ref{tb2}, which shows both the UL and the regularized SSL approaches achieve similar performance.}
\begin{table}[htb]
 \centering
 \caption{\small Sum rate of diff. approaches in weak interf. scenario.}
 \adjustbox{max width=1\columnwidth}{%
 	\begin{tabular}{|l|l|l|}
 		\hline
 		\diagbox{Method}{User Number}                & $K=5$  & $K=10$ \\ \hline
 		{Regularized SSL} & 2.09 (bits/sec) & 2.60 (bits/sec) \\ \hline
 		Unsupervised    & 2.09 (bits/sec) & 2.64 (bits/sec) \\ \hline
		WMMSE           & 2.06 (bits/sec) & 2.74 (bits/sec) \\ \hline
 \end{tabular}}
 \label{tb2}
\end{table}


{ Next, we gradually increase the number of available labeled samples (high-quality labels) of the regularized SSL approach and compare it with the SL approach in~\eqref{eq:supervised}. Both $10$-user and $20$-user cases are considered, and the results for the strong interference scenario is shown in Fig.~\ref{fig4}. We can observe that slightly increasing labeled samples can significantly improve the performance of the regularized SSL approach. Further, compared with the SL approach with $200,000$ labeled samples, the regularized SSL only requires a few hundreds of labeled samples but achieves better sum rate.} 

	

 \begin{figure}[h]
 	\centering
 	\subfigure[Strong Interf. with $K=10$.]{
 		\begin{minipage}[t]{0.5\linewidth}
 			\centering
 		\includegraphics[width=1.6in]{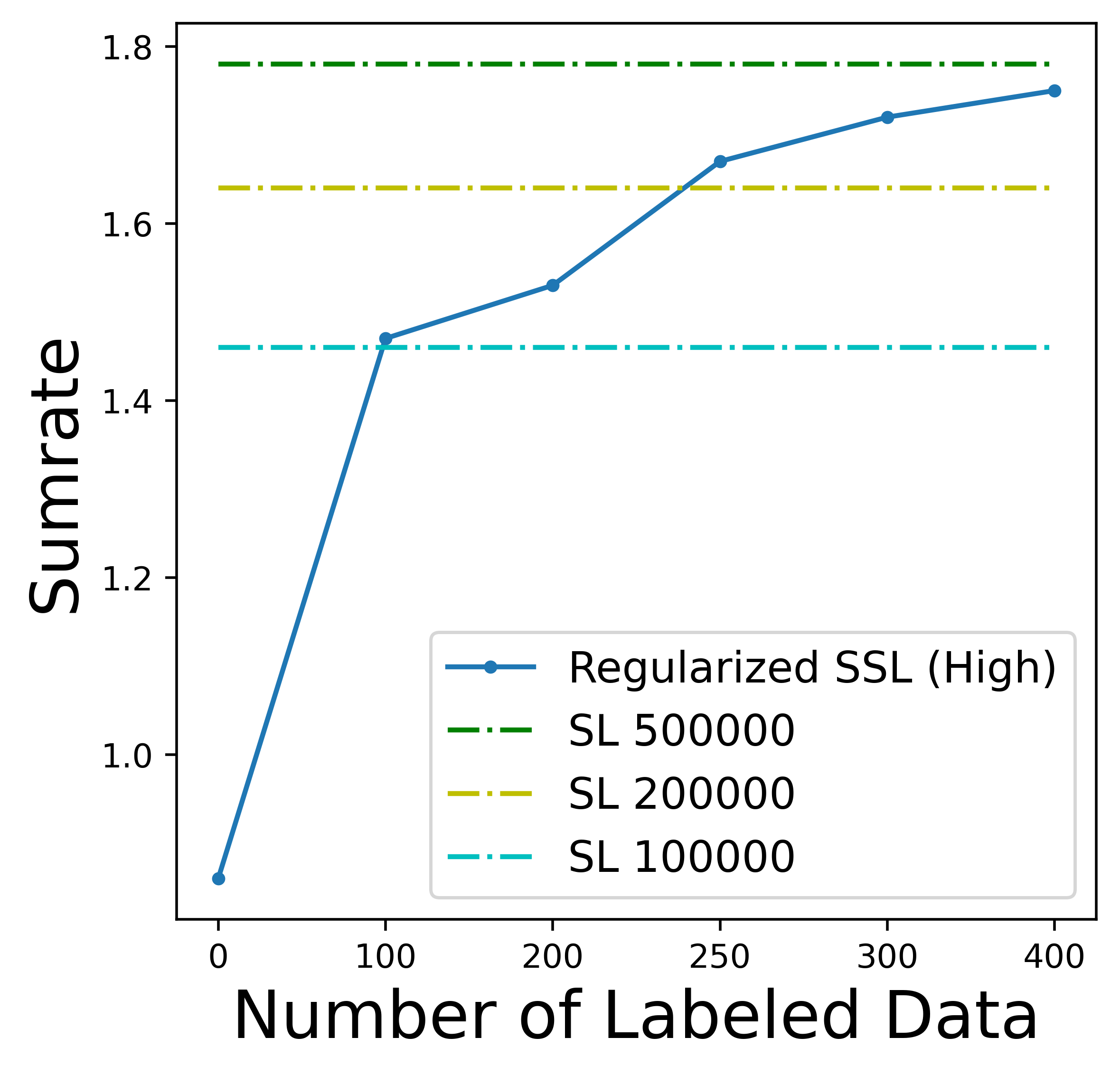}
		\end{minipage}%
 	}%
 	\subfigure[Strong Interf. with $K=20$.]{
 		\begin{minipage}[t]{0.5\linewidth}
 			\centering
 			\includegraphics[width=1.6in]{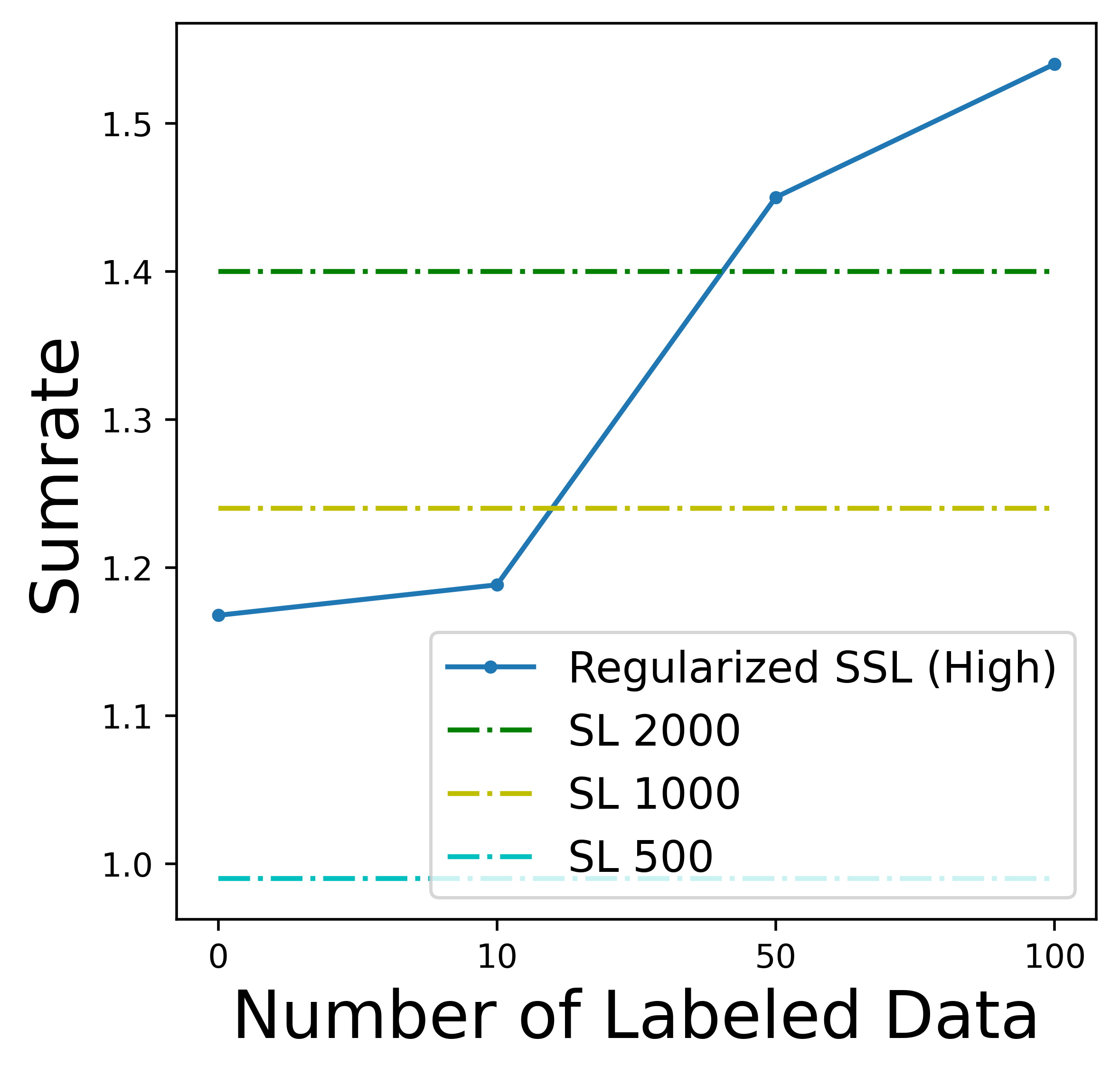}
		\end{minipage}%
 	}%
 	\caption{Sum rate of regularized SSL with diff. number of labeled samples.} 
 	\label{fig4}
 \end{figure}

	
	\section{Conclusion}
	This work analyzes the SL  and UL approaches for learning communication systems. It is shown that under certain conditions (such as having access to high-quality labels), SL can exhibit better convergence properties than UL. To our knowledge, this is the first work that rigorously analyzes the relation between these two approaches. While finding high-quality labels is  challenging, we design a proper semi-supervised learning strategy that only requires a few high-quality labels, but still achieves comparable performance.
	
	
	
\bibliographystyle{IEEEtran}
\bibliography{IEEEabrv,IEEEexample}
	

\appendices
\section{Proof for Claim \ref{claim:2x2}}\label{app:claim:2x2}

\begin{proof}[Proof]
{ 
\underline{\bf The UL problem \eqref{eq:un-supervised}}. 
Let us construct channel samples $|\bh^{(1)}|$ and $|\bh^{(2)}|$ in the  following way:
\begin{equation}\label{eq:example}
\begin{aligned}
|h_{12}^{(1)}|&=|h_{21}^{(1)}|\gg |h_{22}^{(1)}|>|h_{11}^{(1)}|,\\ 
|h_{12}^{(2)}|&=|h_{21}^{(2)}|\gg |h_{11}^{(2)}|>|h_{22}^{(2)}|.
\end{aligned}
\end{equation}
It is straightforward to verify that the optimal solution of Problem~\eqref{eq:wsr} with instance snapshots $\mathbf{h}^{(1)}$ and $\mathbf{h}^{(2)}$ are $\bar{\mathbf{p}}^{(1)}=(0,1)$ and $\bar{\mathbf{p}}^{(2)}=(1,0)$ respectively~\cite{Gjendemsjo06,Charafeddine09}. Further, suppose that the cross channels $h_{12}^{(n)}$ are strong enough such that the following inequality holds
\begin{align}\label{eq:key:construction}
    \frac{2(2+h_{11}^{(n)})|h_{22}^{(n)}|^2}{|h_{11}^{(n)}|^{2}|h_{12}^{(n)}|^{2}}<1.
\end{align}
Also, for notation simplicity, define the following short-handed notations:
\begin{align*}
&\mathbf{P}\left(\bTheta,|\mathbf{H}|\right)  : = \left[\mathbf{p}\left(\bTheta,\left|\mathbf{h}^{(1)}\right|\right);\mathbf{p}\left(\bTheta,\left|\mathbf{h}^{(2)}\right|\right)\right]= [p_1^{(1)};p_2^{(1)};p_1^{(2)};p_2^{(2)}],\\
   &\mathbf{P}^{*}   :=\mathbf{P}\left(\bTheta_{\text {local }},|\mathbf{H}|\right) =\left[p_{1}^{(1),*}; p_{2}^{(1),*}; p_{1}^{(2),*}; p_{2}^{(2),*}\right]=[1; 0; 1; 0].    
\end{align*}

\noindent {\bf Proof of Step 1:} Showing \eqref{eq:key:fx} holds true is equivalent to proving $\mathbf{P}^{*}=[1; 0; 1;0]$ is a local minimum. We prove this by using  contradiction.
Suppose $\mathbf{P}^{*}$ is not a local minimum, then there exists a feasible $\mathbf{P}\neq \mathbf{P}^{*}$ such that for any $\hat{\mathbf{P}}$ between $\mathbf{P}$ and $\mathbf{P}^*$ inequality $\tilde{f}_{\rm ul}(\hat{\mathbf{P}})<\tilde{f}_{\rm ul}(\mathbf{P}^*)$ always holds. Then, by the continuity of $\widetilde{f}_{\text {ul }}(\mathbf{P})$, we can chose $\mathbf{P}$ sufficiently close to $\mathbf{P}^*$ such that the sign of each component of $\nabla_{\mathbf{P}} \widetilde{f}_{\text {ul }}(\mathbf{P})$ keeps the same as $\nabla_{\mathbf{P}} \widetilde{f}_{\text {ul }}(\hat{\mathbf{P}})$, where $\nabla_{\mathbf{P}} \widetilde{f}_{\text {ul }}(\mathbf{P})$ is the gradient of loss $\widetilde{f}_{\text {ul }}(\mathbf{P})$, defined as
\begin{equation*}
\nabla_{\mathbf{P}} \widetilde{f}_{\text {ul }}(\mathbf{P})=\left(\frac{\partial \widetilde{f}_{\text {ul }}}{\partial p_{1}^{(1)}};\frac{\partial \widetilde{f}_{\text {ul }}}{\partial p_{2}^{(1)}};\frac{\partial \widetilde{f}_{\text {ul }}}{\partial p_{1}^{(2)}}; \frac{\partial \widetilde{f}_{\text {ul }}}{\partial p_{2}^{(2)}}\right).
\end{equation*}
According to the Mean Value Theorem, there exists a feasible $\hat{\mathbf{P}}$ between $\mathbf{P}$ and $\mathbf{P}^*$ that 
\begin{equation}\label{eq:contradiction}
\widetilde{f}_{\text {ul }}(\mathbf{P}^*)-\widetilde{f}_{\text {ul }}(\mathbf{P}) = \langle\nabla_{\mathbf{P}} \tilde{f}_{\text {ul }}(\hat{\mathbf{P}}),\mathbf{P}^*-\mathbf{P}\rangle {>}0.
\end{equation}
Note that the feasibility of $\mathbf{P}$ together with $P_{\max}=1$ implies $p_1^{(n)}\leq p_1^{(n),*}=1,\ p_2^{(n)}\geq p_2^{(n),*}=0,\,n=1,2$. Once we show 
$$\frac{\partial \widetilde{f}_{\text {ul }}}{\partial p_{1}^{(n)}}<0,\ \frac{\partial \widetilde{f}_{\text {ul }}}{\partial p_{2}^{(n)}}> 0,\ n=1,2,$$
then $\left\langle\nabla_{\mathbf{P}} \widetilde{f}_{\text {ul }}(\hat{\mathbf{P}}), \mathbf{P}^{*}-\mathbf{P}\right\rangle< 0$ always holds. This contradicts to~\eqref{eq:contradiction}. Next, we show such a contradiction.



Based on our channel construction, $\widetilde{f}_{\rm ul}(\mathbf{P})$ in~\eqref{eq:un-supervised} becomes: 
\begin{align*}
     \widetilde{f}_{\rm ul}(\mathbf{P})=-\sum_{n=1}^{2}\sum_{k=1}^{2} \log \left(1+\frac{\left|h_{kk}^{(n)}\right|^{2} p_{k}^{(n)}}{\sum_{j \neq k}\left|h_{k j}^{(n)}\right|^{2} p_{j}^{(n)}+1}\right)
\end{align*}
The partial derivatives of $\widetilde{f}_{\rm ul}(\mathbf{P})$ are given by:
\begin{align}\label{eq:f:derivative}
\scriptsize
\left\{\begin{array}{l}
\begin{aligned}
&\frac{\partial \widetilde{f}_{\rm ul}}{\partial p_1^{(n)}}= -\frac{|{h_{11}^{(n)}}|^2}{|{h_{11}^{(n)}}|^2p_1^{(n)}+|{h_{12}^{(n)}}|^2 p_2^{(n)}+1}\\
&\quad \quad+\frac{|{h_{21}^{(n)}}|^2 |{h_{22}^{(n)}}|^2 p_2^{(n)}}{\left(|{h_{21}^{(n)}}|^2 p_1^{(n)}+|{h_{22}^{(n)}}|^2 p_2^{(n)}+1\right)\left(|{h_{21}^{(n)}}|^2p_1^{(n)}+1\right)}
\end{aligned}\\
\begin{aligned}
&\frac{\partial \widetilde{f}_{\rm ul}}{\partial p_2^{(n)}}=-\frac{|{h_{22}^{(n)}}|^2}{|{h_{21}^{(n)}}|^2p_1^{(n)}+|{h_{22}^{(n)}}|^2 p_2^{(n)}+1}\\
&\quad \quad+\frac{|{h_{11}^{(n)}}|^2 |{h_{12}^{(n)}}|^2 p_1^{(n)}}{\left(|{h_{12}^{(n)}}|^2 p_2^{(n)}+|{h_{11}^{(n)}}|^2 p_1^{(n)}+1\right)\left(|{h_{12}^{(n)}}|^2p_2^{(n)}+1\right)}
\end{aligned}
\end{array}\right..
\end{align}
Combining with~\eqref{eq:key:construction}, it is straightforward to verify 
\begin{equation}\label{eq:region}
\begin{aligned}
   &\frac{2(2+h_{11}^{(n)})|h_{22}^{(n)}|^2}{|h_{11}^{(n)}|^{2}|h_{12}^{(n)}|^{2}}< p_1^{(n)} \leq 1 ,\\
    &0\leq  p_2^{(n)} < \operatorname{min}\left\{c_1^{(n)},c_2^{(n)}\right\}, \quad \forall~n=1,2,
\end{aligned}
\end{equation}
where $c_1^{(n)} = (|h_{11}^{(n)}|^2)/((|h_{11}^{(n)}|^2+|h_{12}^{(n)}|^2+1)|h_{21}^{(n)}|^2|h_{22}^{(n)}|^2)$ and $c_2^{(n)} = \frac{1}{|h_{12}^{(n)}|^2}$.
Based on relations in~\eqref{eq:region}, we can show that the gradient expression \eqref{eq:f:derivative} satisfies the following:
\begin{align*}
     \frac{\partial \widetilde{f}_{\text {ul }}}{\partial p_{1}^{(n)}} &\leq -\frac{|h_{11}^{(n)}|^{2}}{|h_{11}^{(n)}|^{2} +|h_{12}^{(n)}|^{2}+1} + |h_{21}^{(n)}|^{2}|h_{22}^{(n)}|^{2} p_{2}^{(n)}\\
     &< -\frac{|h_{11}^{(n)}|^{2}}{|h_{11}^{(n)}|^{2} +|h_{12}^{(n)}|^{2}+1}+|h_{21}^{(n)}|^{2}|h_{22}^{(n)}|^{2}\\
     &\quad\times\frac{|h_{11}^{(n)}|^2}{(|h_{11}^{(n)}|^2+|h_{12}^{(n)}|^2+1)|h_{21}^{(n)}|^2|h_{22}^{(n)}|^2}<0\nonumber\\
    \frac{\partial \widetilde{f}_{\text {ul}}}{\partial p_{2}^{(n)}}&\geq -\frac{|h_{22}^{(n)}|^{2}}{|h_{21}^{(n)}|^{2} p_{1}^{(n)}+1}\\
    &\quad+\frac{|h_{11}^{(n)}|^{2}|h_{12}^{(n)}|^{2} p_{1}^{(n)}}{\left(|h_{12}^{(n)}|^{2} p_{2}^{(n)}+|h_{11}^{(n)}|^{2} +1\right)\left(|h_{12}^{(n)}|^{2} p_{2}^{(n)}+1\right)}\\
    &> -|h_{22}^{(n)}|^{2}+\frac{|h_{11}^{(n)}|^{2}|h_{12}^{(n)}|^{2} p_{1}^{(n)}}{2\left( |h_{11}^{(n)}|^{2}+2\right)} >0.
\end{align*}
This is a contradiction to \eqref{eq:contradiction}. Hence $\mathbf{P}^*$ is a local minimum and there exists a region $N_{\epsilon}(\mathbf{P}^*)$ with~\eqref{eq:key:fx} holds.

\noindent{\bf Proof of Step 2:} Next we show that for every $\widetilde{\bTheta}$ satisfying   $\mathbf{P}\left(\widetilde{\bTheta},|\mathbf{h}|\right)=\mathbf{P}^{*}$, there exists a region $N_{\delta}(\widetilde{\bTheta})$ such that for all $\bTheta \in N_{\delta}(\widetilde{\bTheta})$, the corresponding $\mathbf{P}(\bTheta,|\mathbf{H}|)$ falls in $N_{\epsilon}\left(\mathbf{P}^{*}\right)$.

Let us fix a $\widetilde{\bTheta}$ satisfying $\mathbf{P}^{*}=\widetilde{\bTheta}|\mathbf{H}|$ (linear neural network). Then for the constant $\epsilon>0$ identified in Step 1, define $\delta=\epsilon/\| |\mathbf{H}| \|$ and the region $N_{\delta}(\widetilde{\bTheta})$ as 
$$
N_{\delta}(\widetilde{\bTheta}):=\{ \bTheta \mid \|\bTheta-\widetilde{\bTheta}\|\leq \delta, \mathbf{0}\le \bTheta|\mathbf{H}|\le \mathbf{1} \}.
$$
For any $\bTheta\in N_{\delta}(\widetilde{\bTheta})$, we have
$$\|\mathbf{P}(\bTheta,|\mathbf{H}|)-\mathbf{P}(\widetilde{\bTheta},|\mathbf{H}|)\|\leq \| \bTheta - \widetilde{\bTheta} \| \| |\mathbf{H}|\|  \leq \epsilon,$$
where the Cauchy-Schwarz inequality is used.
This implies $\mathbf{P}(\bTheta,|\mathbf{H}|)\in N_{\epsilon}\left(\mathbf{P}^{*}\right)$ and hence 
$$
f_{\rm ul}({\bTheta})=\widetilde{f}_{\rm ul}(\mathbf{P}(\bTheta,|\mathbf{H}|))>\widetilde{f}_{\rm ul}(\mathbf{P}^{*})=f_{\rm ul}(\widetilde{\bTheta}),\ \forall {\bTheta}\in N_{\delta}(\widetilde{\bTheta}).
$$
By optimality definition, such a $\widetilde{\bTheta}$ is a local minimum of \eqref{eq:un-supervised}.

\noindent\underline{\bf The SL problem \eqref{eq:supervised}.} 
${f}_{\rm sl}(\bTheta)$ is a convex function given as 
\vspace{-0.5em}
\begin{equation*}
f_{\rm sl}(\bTheta)=\sum_{n=1}^{2}\|\bTheta|\bh^{(n)}|-\bar{\mathbf{p}}^{(n)}\|^2.
\end{equation*}
Since $\bTheta$ contains 8 scalar parameters, minimizing $f_{\rm sl}(\bTheta)$ is equivalent to  solving the following linear equations, $$\bTheta |\mathbf{h}^{(1)}| = [0; 1],  \quad \bTheta |\mathbf{h}^{(2)}| = [1; 0].$$ 
It follows that as long as the channel realizations are randomly generated so that they are linearly independent, there always exists $\bTheta$ which can predict the true labels. Hence, any optimal solution achieves zero loss.
 
 
     

}
 \end{proof}

\section{Proof of Claim {2} }\label{app:claim:zero:loss}
\begin{proof}[Proof]
{ 
The main idea of the proof is based on analyzing the relation between KKT conditions of problems~\eqref{eq:supervised} and \eqref{eq:un-supervised}. We will show that each KKT solution of problem~\eqref{eq:supervised} is also a KKT solution of problem \eqref{eq:un-supervised} but not the other way.


The KKT condition for the UL problem~\eqref{eq:un-supervised} is that, there is a tuple $(\tilde{\bTheta},\tilde{\boldsymbol{\lambda}},\tilde{\boldsymbol{\mu}})$ such that the following relations hold: 
{\small 
\begin{subequations}\label{KKT:UL:pf}
\begin{align}
&\sum\nolimits_{n,k} \left(-\nabla R_k^{(n)}-\tilde{\lambda}_{k}^{(n)}+\tilde{\mu}_{k}^{(n)} \right) \cdot \nabla p_k^{(n)} = \bf 0,\label{KKT:ULa}\\
&\tilde{\lambda}_{k}^{(n)} \geq 0, \tilde{\lambda}_{k}^{(n)} \cdot p_k(\tilde{\bTheta};|\bh^{(n)}|) =0,\ \forall n,k, \label{KKT:ULb}\\ &\tilde{\mu}_{k}^{(n)}  \geq 0,\tilde{\mu}_{k}^{(n)} \cdot \left(p_k(\tilde{\bTheta};|\bh^{(n)}|)-P_{\rm max}\right)=0,\ \forall n,k, \label{KKT:ULc}\\
&{\bf 0}   \leq \mathbf{P}(\tilde{\bTheta};|\mathbf{H}|)\leq {\bf P_{\rm max}}\label{KKT:ULd},
\end{align}
\end{subequations}}
where 
$${\small\nabla R_k^{(n)}:=\frac{\partial R\left(\mathbf{p}\left(\tilde{\bTheta} ;|\mathbf{h}^{(n)}|\right);|\mathbf{h}^{(n)}|\right)}{\partial p_k},\  \nabla p_k^{(n)}:=\frac{\partial p_k(\tilde{\bTheta};|\bh^{(n)}|)}{\partial \tilde{\bTheta}}}.$$
By stationary solution assumption for $\bar{\bp}^{(n)},n \in [N]$, i.e., $\bar{\bp}^{(n)}$ is a stationary solution of WSR problem \eqref{eq:wsr}, there must exist a tuple $(\bar{\bp}^{(n)},\bar{\boldsymbol{\lambda}}^{(n)},\bar{\boldsymbol{\mu}}^{(n)})$ such that for all $n\in[N]$ the following holds true:
{\small\begin{subequations}\label{eq:KKT}
\begin{align}
&-\partial R(\bar{\mathbf p}^{(n)},|\mathbf{h}^{(n)}|)/\partial p_k -  \bar{\lambda}_{k}^{(n)}+ 
\bar{\mu}_{k}^{(n)}=0,\ \forall k,\label{eq:KKTa}\\
&\bar{\lambda}_{k}^{(n)} \geq 0, \;\bar{\lambda}_{k}^{(n)} \cdot p_k^{(n)} =0,\ \forall k, \label{eq:KKTb}\\ 
&\bar{\mu}_{k}^{(n)}  \geq 0,\;
\;\bar{\mu}_{k}^{(n)} \cdot \left(p_k^{(n)}-P_{\rm max}\right)=0,\ \forall k,\label{eq:KKTc}\\
& 0\leq p_k^{(n)}\leq P_{\rm max},\ \forall k.\label{eq:KKTd}
\end{align}
\end{subequations}}
Now we argue that the tuple $(\bTheta^{*},\bar{\boldsymbol{\lambda}},\bar{\boldsymbol{\mu}})$ (with $\bTheta^{*}\in\mathcal{S}$) satisfies the KKT conditions in \eqref{KKT:UL:pf}. By zero loss condition for $\bTheta^{*}$, the following holds:
\begin{align}\label{eq:zero:loss:pf}
   p_k(\bTheta^{*}(\bar{\bp});|\mathbf{h}^{(n)}|)  =\bar{p}_k^{(n)}, \quad \forall~n \in [N], \; \forall~ k\in [K].
\end{align}
\eqref{eq:zero:loss:pf} together with~\eqref{eq:KKTb}-\eqref{eq:KKTd} immediately imply \eqref{KKT:ULb}-\eqref{KKT:ULd} hold with $\tilde{\bTheta}=\bTheta^{*},\tilde{\lambda}_{k}^{(n)}=\bar{\lambda}_{k}^{(n)},\tilde{\mu}_{k}^{(n)}=\bar{\mu}_{k}^{(n)},\forall n,k$. It remains to verify~\eqref{KKT:ULa} holds. Denote 
\begin{align*}
    \small
    \nabla R_k^{*,(n)}&:=\frac{\partial R\left(\mathbf{p}\left(\bTheta^{*} ;|\mathbf{h}^{(n)}|\right);|\mathbf{h}^{(n)}|\right)}{\partial p_k},\\
    \nabla p_k^{*,(n)}&:=\frac{\partial p_k(\bTheta^{*};|\bh^{(n)}|)}{\partial \bTheta^{*}},
\end{align*}
then by \eqref{eq:zero:loss:pf}, \eqref{KKT:ULa} can be expressed as: 
{\small
\begin{align*}
    &\sum\nolimits_{n,k} \bigg(-\nabla R_k^{*,(n)} - \bar{\lambda}_{k}^{(n)}+  \bar{\mu}_{k}^{(n)}\bigg)\nabla p_k^{*,(n)}\\
    &= \sum\nolimits_{n,k}\bigg(-\frac{ \partial R(\bar{\mathbf p}^{(n)},|\mathbf{h}^{(n)}|)}{\partial p_k} -  \bar{\lambda}_{k}^{(n)}+ 
    \bar{\mu}_{k}^{(n)}\bigg)\nabla p_k^{*,(n)}=\bf{0},
\end{align*}}
where first equality comes from~\eqref{eq:KKTa}.

Finally, it is easy to see that there may exists a solution in $\mathcal{B}$ that is not an optimal solution for \eqref{eq:supervised} (c.f. the example constructed in Claim \ref{claim:2x2}). This is because \eqref{KKT:ULa} in general can not guarantee \eqref{eq:KKTa} holds simultaneously for all $n\in [N]$. This completes the proof.

}
\end{proof}



{
\section{$\Lambda_1$ and $\Lambda_2$}\label{app:Lambda}
\noindent\textbf{Expression of $\Lambda_1$ and $\Lambda_2$:} Recall Section\ref{sec:assumption} and  Then $\Lambda_1$ and $\Lambda_2$ are defined as 
\begin{align*}
&\Lambda_1:=
\bigg(\frac{\gamma^{4}}{3}\left(\frac{6}{\gamma^{2}}\right)^{L}\|\mathbf{H}\|_{F} \sqrt{2 f_{\rm SL}\left(\bTheta^{0}\right)}\cdot  \frac{\bar{\lambda}_{3 \rightarrow L}}{(\underline{\lambda}_{3 \rightarrow L})^{2}}\\\nonumber
 &\times e^\frac{{2\left({\frac{3}{2}}\right)^L\cdot{\|\mathbf{H}\|_F\prod_{l=1}^{L}\bar{\lambda}_l}}}{\alpha}\cdot\max \left(\frac{2 \bar{\lambda}_{1}\cdot  \bar{\lambda}_{2}}{\min _{l \in\{3, \ldots, L\}} \underline{\lambda}_{l}\cdot  \bar{\lambda}_{l}}, \bar{\lambda}_{1}, \bar{\lambda}_{2}\right)\bigg)^\frac{1}{2}\\
&\Lambda_2:= \bigg(\frac{2 \gamma^{4}}{3}\left(\frac{6}{\gamma^{2}}\right)^{L}\sigma_{\max }(\mathbf{H})\|\mathbf{H}\|_{F}\times e^{\frac{{2\left({\frac{3}{2}}\right)^L\cdot{\|\mathbf{H}\|_F\prod_{l=1}^{L}\bar{\lambda}_l}}}{\alpha} }\nonumber\\
& \quad \times \sqrt{2 f_{\rm SL}\left(\bTheta^{0}\right)}\cdot  \frac{\bar{\lambda}_{3 \rightarrow L}}{({\underline{\lambda}_{3 \rightarrow L}})^2} \cdot \bar{\lambda}_{2}\bigg)^\frac{1}{3}.\nonumber
\end{align*}

\noindent\textbf{Discussion on Condition~\eqref{ass:square}:} Next, let us the condition~\eqref{ass:square} in Assumption~\ref{ass:init:sl}. In order to satisfy it, we can use the following initialization strategy~\cite{nguyen2020global}. First, initialize $
\left[\mathbf{W}_{1}^{0}\right]_{i j} \sim \mathcal{N}(0,1 / K^2)
$, thus $\underline{\lambda}_{H}$ is strictly positive with probability $1$. Next we will show that $\Lambda_1$ can be made arbitrarily small. Pick
$\left(\mathbf{W}_{l}^{0}\right)_{l=3}^{L}$ such that $\underline{\lambda}_{l} \geq 1 \text { and }(\underline{\lambda}_{l})^{2} \geq c \bar{\lambda}_{l}$, where $c>1$ for $l=3,\cdots,L$. One example is to choose $\mathbf{W}_{l}^{0}$'s as scaled identity matrices, whose top block is scaled identity, that is:
\begin{equation}
    \mathbf{W}_l^{0} = \begin{bmatrix}
   \begin{array}{c}
   c \cdot \mathbf{I}_{n_{l}}  \\ 
   \mathbf{0} 
   \end{array}
 \end{bmatrix} \in \mathbb{R}^{n_{l}\times n_{l-1}}, \; l=3, \cdots, L.
\end{equation}
Then we can upper bound $\bar{\lambda}_2$ and $ f_{\rm SL}(\bTheta^0)$ by constants.
Set $\left[\mathbf{W}_{2}^{0}\right]_{i j} \sim \mathcal{N}(0,v)
$. If $v$ is small enough, then we can upper bound $\bar{\lambda}_2$ with high probability:
\begin{equation}
   \bar{\lambda}_{2} =\frac{2}{3}\cdot\left(1+\sigma_{\max }(W_{2}^{0})\right) \leq 1. \nonumber
\end{equation}
To bound $f_{\rm SL}(\bTheta^0)$, recall the notation in Section~\ref{sec:preliminaries}, we have:
\begin{align}
  \sqrt{2 f_{\rm SL}\left(\bTheta^{0}\right)}  
  &=\|\bP-\bar{\bP}\|_F=\|\mathbf{f}_L-\mathbf{y}\|_2\\\nonumber
  &\leq\|\mathbf{y}\|_{2}+\left\|\mathbf{F}_{L}\left(\bTheta^{0}\right)\right\|_{F}\\\nonumber &\leq\|\mathbf{y}\|_{2}+\prod_{l=1}^{L}\sigma_{\max }(\mathbf{W}_{l}^{0})\|\mathbf{H}\|_{F}. \nonumber 
\end{align}
If $v$ is small enough, then  
\begin{align*}
\|\mathbf{y}\|_{2}+\prod_{l=1}^{L}\sigma_{\max }(\mathbf{W}_{l}^{0})\|\mathbf{H}\|_{F} \leq 2\|\mathbf{y}\|_2
\end{align*}
holds with high probability and hence $\sqrt{2 f_{\rm SL}\left(\bTheta^{0}\right)} \leq 2\|\mathbf{y}\|_2$. 

    
Summarizing the above, \eqref{ass:square} holds with the following sufficient condition:
\begin{align*}
&(\underline{\lambda}_{H})^{2}\cdot\left(\frac{\gamma^{4}}{3}\left(\frac{6}{\gamma^{2}}\right)^{L} 2\|\mathbf{H}\|_{F}\|\mathbf{y}\|_{2}\cdot \max \left(2 \bar{\lambda}_{1}, 1\right)\right)^{-1}\\
&\quad\geq \frac{\bar{\lambda}_{3 \rightarrow L}}{(\underline{\lambda}_{3 \rightarrow L})^{2}}\times e^{\frac{{2\left({\frac{3}{2}}\right)^L\cdot{\|\mathbf{H}\|_F\prod_{l=1}^{L}\bar{\lambda}_l}}}{\alpha} }\\
&\quad \geq (\frac{1}{c})^{L-2}\cdot e^{\frac{{2\left({\frac{3}{2}}\right)^L\cdot{\|\mathbf{H}\|_F\prod_{l=1}^{L}\bar{\lambda}_l}}}{\alpha} }.
\end{align*}
Notice that $\underline{\lambda}_H$ is only dependent on $\mathbf{W}_1^0$ and $\mathbf{H}$, while $\Lambda_1$ depends on the weights of rest layers. The initialization requires that $\underline{\lambda}_H$ is a fixed positive number, while changing the parameters can make $\Lambda_1$ as arbitrarily small. Specifically, this can be done by increasing $c$ and choose a large enough $\alpha$ dependent on $\bar{\lambda}_l$. Thus, there exist large $c$ and small $v$ to satisfy \eqref{ass:square}. Similarly we can show $\Lambda_2$ can be made arbitrarily small.
}

\section{Proof of Claim \ref{thm:zero:loss}}\label{app:thm:zero:loss}
{
Let us re-state our objective function and introduce some useful notations. {Recall notations in~\eqref{eq:nn}, that we define $\mathbf{f}_l = \operatorname{vec}(\mathbf{F}_l),\ \mathbf{w}_l = \operatorname{vec}(\mathbf{W}_l),\ l \in [L]$, and $\mathbf{y} = \operatorname{vec}(\bar{\mathbf{P}})$, which represents the vectorized output, vectorized weight, and the vectorized label, respectively.} 
For each layer $l,\ l\in[L-1]$, let us use $\mathbf{\Sigma}_l$ to represent the derivative of $a\left(\mathbf{F}_{l-1}\mathbf{W}_l\right)$ w.r.t. $\mathbf{W}_l$ (and similarly for the last layer with activation $b(\cdot)$),
\begin{equation}\label{eq:sigma}
\begin{aligned}
  &\mathbf{\Sigma}_{l}:=\operatorname{diag}\left[\operatorname{vec}\left(a^{\prime}\left(\mathbf{F}_{l-1}\mathbf{W}_l\right)\right)\right] \in \mathbb{R}^{N n_{l} \times N n_{l}},\ l \in [L-1],\\ 
  &\mathbf{\Sigma}_{L}  :=\operatorname{diag}\left[\operatorname{vec}\left(b^{\prime}\left(\mathbf{F}_{L-1}\mathbf{W}_L\right)\right)\right] \in \mathbb{R}^{N n_{L} \times N n_{L}}.  
\end{aligned}
\end{equation}
At the $m$th training iteration with $\bTheta^m=\{\mathbf{W}_{l}^m\}_{l=1}^{L}$, denote $\mathbf{F}_{l}^{m}$ as the output of the $l$-th layer for all samples. 
Note that the last layer's output $\mathbf{f}_L$ is a vector function of $\bTheta$, we further denote its Jacobian matrix as:
\begin{align}
   &\mathbf{J}_{L}=\left[\frac{\partial \mathbf{f}_{L}}{\partial \mathbf{w}_1}, \ldots, \frac{\partial \mathbf{f}_{L}}{\partial \mathbf{w}_L}\right],\ \frac{\partial \mathbf{f}_{L}}{\partial \mathbf{w}_l} \in \mathbb{R}^{\left(N n_{L}\right) \times\left(n_{l-1} n_{l}\right)},\ l \in[L]. \nonumber
\end{align}
We use $\nabla_l f_{\rm SL}$ and $\nabla_l f_{\rm UL}$ to denote the partial gradient $f_{\rm SL}$ and $f_{\rm UL}$ w.r.t. $\mathbf{w}_l$ (vectorized weight of the $l$th layer) respectively. Finally, the concatenated allocated power in training is denoted as ${\mathbf{q}}:=(\mathbf{q}^{(1)};\mathbf{q}^{(2)};\cdots;\mathbf{q}^{(N)})$, where each $\mathbf{q}^{(n)}={\mathbf q}(\mathbf \Theta; |\mathbf h^{(n)}|)$ is as
{\small 
\begin{align*}
\frac{\partial {f}_{\rm SL}(\mathbf{q}(\bTheta))}{\text{vec}({\partial {\mathbf{q}}})}&=\mathbf{q}-\mathbf{y}=\mathbf{f}_L-\mathbf{y},\\
\frac{\partial {f}_{\rm UL}(\mathbf{q}(\bTheta))}{{\partial \text{vec}({\mathbf{q}}})}&=(\nabla R_1^{(1)},\nabla R_2^{(1)},\ldots,\nabla R_1^{(N)},\ldots,\nabla R_K^{(N)})^T:=\mathbf{\nabla r},
\end{align*}
}
where $\small\nabla R_k^{(n)}:=-\frac{\partial R\left(\mathbf{q}\left({\bTheta} ;|\mathbf{h}^{(n)}|\right);|\mathbf{h}^{(n)}|\right)}{\partial q_k^{(n)}},\forall k,n$.

Based on these notation, two key lemmas are given below.

\begin{lemma}\label{lemma1}
For $l\in[L]$ the following holds:
\begin{align}
&\nabla_{{l}} f_{\rm SL} =\left(\mathbf{I}_{n_{l}} \otimes \mathbf{F}_{l-1}^{T}\right) \prod_{t=l+1}^{L} \mathbf{\Sigma}_{t-1}\left(\mathbf{W}_{t} \otimes \mathbf{I}_{N}\right)\mathbf{\Sigma}_L(\mathbf{f}_L-\mathbf{y}),\nonumber\\
&\nabla_{{l}} f_{\rm UL} =\left(\mathbf{I}_{n_{l}} \otimes \mathbf{F}_{l-1}^{T}\right) \prod_{t=l+1}^{L} \mathbf{\Sigma}_{t-1}\left(\mathbf{W}_{t} \otimes \mathbf{I}_{N}\right)\mathbf{\Sigma}_L \mathbf{\nabla r}, \nonumber\\
&\frac{\partial \mathbf{f}_{L}}{\partial \mathbf{w}_l}=\mathbf{\Sigma}_L \prod_{t=0}^{L-l-1}\left(\mathbf{W}_{L-t}^{T} \otimes \mathbf{I}_{N}\right) \mathbf{\Sigma}_{L-t-1}\left(\mathbf{I}_{n_{l}} \otimes \mathbf{F}_{l-1}\right) \nonumber.
\end{align}
\end{lemma}
The above lemma provides expressions of the gradient of both objective functions as well as the Jacobian matrix. The gradient of SL is from \cite[Lemma 4.1]{nguyen2020global} while the gradient of UL is slightly modified because the objective function is different.

\begin{lemma}\label{lemma2}
Suppose Assumptions~\ref{ass:width} and~\ref{ass:activation} hold. Then, for any $\mathbf{W}_{l},\ l\in[L]$, the following holds
{\small
\begin{align}
\left\|\mathbf{F}_{l}\right\|_{F}  &\leq\|\mathbf{H}\|_{F} \prod_{t=1}^{l}\sigma_{\max}(\mathbf{W}_{t}),\ ~l \in[L-1],\label{eq:upp1}\\
\left\|\mathbf{F}_{L}\right\|_{F}  &\leq (1+\alpha)\sqrt{Nn_L},\label{eq:upp2}\\
\left\|\nabla_{{l}} f_{\rm SL}\right\|_2 &\leq\|\mathbf{H}\|_{F} \prod_{t=1, t \neq l}^{L}\sigma_{\max}(\mathbf{W}_{t})\|\mathbf{f}_L-\mathbf{y}\|_{2}, \ ~l\in [L], \label{SLgradbound}\\
\left\|\nabla_{{l}} f_{\rm UL}\right\|_2 &\leq\|\mathbf{H}\|_{F} \prod_{t=1, t  \neq l}^{L}\sigma_{\max}(\mathbf{W}_{t})\left\| \mathbf{\nabla r}\right\|_{2}, \ l\in [L].\label{ULgradbound}
\end{align}
}
Furthermore, denote $\bTheta=\left(\mathbf{W}_{l}\right)_{l=1}^{L}, \bTheta^\prime=\left(\mathbf{W}_{l}^\prime\right)_{l=1}^{L}$, then for $l \in[L]$, the following inequalities hold,
\begin{equation}\label{eq:lipconstant}
\begin{aligned}
\left\|\mathbf{F}_{L}-\mathbf{F}_{L}^\prime\right\|_{F} &\leq c_1 \left\|\bTheta-\bTheta^\prime\right\|_{F},\\ \left\|\frac{\partial \mathbf{f}_{L}}{\partial \mathbf{w}_l}-\frac{ \partial \mathbf{f}_{L}^\prime}{\partial \mathbf{w}_l^\prime}\right\|_{2}\nonumber &\leq c_2\left\|\bTheta-\bTheta^\prime\right\|_{F}.
\end{aligned}
\end{equation}
where 
{\small 
\begin{align*}
&c_1 = \sqrt{LNn_L}\|\mathbf{H}\|_{F} \frac{\prod_{l=1}^{L} \bar{\lambda}_{l}}{\min _{l \in[L]} \bar{\lambda}_{l}},\ 
c_2=\sqrt{L}\|\mathbf{H}\|_{F} R\left(1+(L \beta + 1) \|\mathbf{H}\|_F R\right),\\
&\bar{\lambda}_{l} = \max \left(\sigma_{\max}(\mathbf{W}_{l}),\sigma_{\max}(\mathbf{W}_{l}^\prime)\right),l\in[L],\ R=\prod_{l=1}^{L} \max \left(1, \bar{\lambda}_{l}\right).
\end{align*}
}
\end{lemma}
First, Lemma~\ref{lemma2} provides an upper bound for the output of each layer (see~\eqref{eq:upp1} and~\eqref{eq:upp2}) as well as for the partial gradient of each layer (see~\eqref{SLgradbound} and~\eqref{ULgradbound}). Second, it shows how smooth the DNN mapping and the corresponding gradient are (see~\eqref{eq:lipconstant}). This smoothness property is important for showing the convergence statement in our claim. 
Note that this lemma is slightly different from \cite[Lemma 4.2]{nguyen2020global} since we need to adapt to the last layer with additional SCReLU activation. The output range is $[-\alpha,1+\alpha]$, so the bound of  ${\|\mathbf{F}_L\|}_F$ in \eqref{eq:upp2} would be different. The upper bound of ${\|\nabla_{l}f_{\rm SL}\|}_2$ in~\eqref{SLgradbound} is directly applied from \cite[Lemma 4.2]{nguyen2020global}, while $\|\nabla_{l}f_{\rm UL}\|_2$ in \eqref{ULgradbound} is slightly modified because the objective is different. Further, in \eqref{eq:lipconstant}, the bound is adapted to our setting with SCReLu activation function in the last layer. 


We are now ready to prove Claim \ref{thm:zero:loss}.

\noindent {\bf Proof of Claim \ref{thm:zero:loss}.} Consider $P_{\max}=\sigma=1$ for simplicity. 

\noindent\underline{\bf SL Training~\eqref{eq:SL:GD}:} Conclusion for SL training in Claim 3 can be derived by slightly modifying from~\cite[Theorem 3.2]{nguyen2020global}, the difference is that we have used SCReLu activation function in the last layer. 
In order to ensure linear convergence, we need to provide a strictly positive lower bound for $\|\frac{\partial \mathbf{f}_{L}}{\partial \mathbf{w}_l }\|_2$, that is, we need to find a constant $\alpha_0$, such that 
\begin{align}\label{eq:lower:bound}
   \left\|\frac{\partial \mathbf{f}_{L}}{\partial\mathbf{w}_l}\right\|_2^2\geq \alpha_0\|\mathbf{f}_L-\mathbf{y}\|_2^2. 
\end{align}
We show that for the derivative of the last layer $\mathbf{\Sigma}_L$, if $\| \mathbf{F}_{L-1}\mathbf{W}_L\|_F$ is bounded, then its smallest singular value is also lower bounded. Thus we can find a positive $\alpha_0$ for~\eqref{eq:lower:bound}.

More specifically, by Lemma \ref{lemma2}, we have:
\begin{equation}
    \|\mathbf{F}_{L-1} \mathbf{W}_{L}\|_F\leq \|\mathbf{H}\|_{F} \prod_{l=1}^{L}\sigma_{\max}(\mathbf{W}_{l}):=A.
\end{equation}
Then denote $\mathbf{G}=\mathbf{F}_{L-1} \mathbf{W}_{L}$, we have $\mathbf{G}_{ij}\in [-A,A]$. Now we write down the derivative of the SCReLU function:
\begin{equation*}
	b^\prime(\mathbf{G}_{ij}) = \left\{\begin{array}{ll}
			e^{\frac{\mathbf{G}_{ij}}{\alpha}}  & \mathbf{G}_{ij} < {0}\\
			1   & \mathbf{G}_{ij}\in [{0},\;{1}]\\
			e^{\frac{1-\mathbf{G}_{ij}}{\alpha}}\; & \mathbf{G}_{ij}>1.
		\end{array}\right.
	 	\end{equation*}
Thus we have 
\begin{equation}\label{derivativebound}
    b^\prime(\mathbf{G}_{ij})\geq e^{-\frac{A}{\alpha}}.
\end{equation}
This implies that:
\begin{align}\label{defineB}
    \sigma_{\min}^2(\mathbf{\Sigma}_{L})& \stackrel{\eqref{eq:sigma}}=\min_{i,j}\{ |b^\prime(\mathbf{G}_{ij})|^2\}
    \stackrel{\eqref{derivativebound}}\geq \left(e^{-\frac{A}{\alpha}}\right)^2:= B.
\end{align}
If we choose $\alpha:=\bar{\alpha}_H>
\left(\frac{3}{2}\right)^{L} \cdot\|\mathbf{H}\|_{F} \prod_{l=1}^{L} \bar{\lambda}_{l}
$, we have: 
\begin{align}
\scriptsize
    \|\nabla_{{2}} f_{\mathrm{SL}}\|_2^2
    &\stackrel{(i)}=\|\mathbf{I}_{n_{2}} \otimes \mathbf{F}_{1}^{T} \prod_{t=3}^{L} \mathbf{\Sigma}_{t-1}\left(\mathbf{W}_{t} \otimes \mathbf{I}_{N}\right) \mathbf{\Sigma}_{L} \cdot (\mathbf{f}_L-\mathbf{y})\| _2^2 \nonumber\\
    &\stackrel{(ii)}\geq \sigma^2_{\min}( \mathbf{F}_1^{T})\cdot\prod_{t=3}^{L}\sigma^2_{\min}(\mathbf{\Sigma}_{t-1}) \sigma^2_{\min}(\mathbf{W}_{t})\nonumber\\
    &\quad \times \sigma^2_{\min}(\mathbf{\Sigma}_L) \times \|\mathbf{f}_L-\mathbf{y}\|_2^2\nonumber\\
    &\stackrel{(iii)}\geq (\frac{1}{2}\underline{\lambda}_H)^2\cdot \prod_{t=3}^{L}\gamma^2(\frac{1}{2}\underline{\lambda}_t)^2 \cdot e^{-2}\cdot\|\mathbf{f}_L-\mathbf{y}\|_2^2\nonumber\\
    &= e^{-2}\gamma^{2(L-2)}\left(\frac{1}{2}\right)^{2(L-1)} (\underline{\lambda}_{3 \rightarrow L})^2(\underline{\lambda}_{H})^{2}\left\|\mathbf{f}_{L}-\mathbf{y}\right\|_{2}^{2}\nonumber
\end{align}
where $(i)$ follows from Lemma \ref{lemma1}; $(ii)$ is by \cite[Lemma 4.1(3)]{nguyen2020global}, which provides a lower bound for $\|\nabla_{{2}} f_{\mathrm{SL}}\|_2^2$; $(iii)$ is because $B$ defined in \eqref{defineB} satisfies $B\leq e^{-2}$ and  the smallest singular value of weight in each layer is lower bounded.
Now we have shown that \eqref{eq:lower:bound} holds true with 
\begin{equation}\label{eq:alpha0}
\alpha_0=e^{-2}\gamma^{2(L-2)}\left(\frac{1}{2}\right)^{2(L-1)} (\underline{\lambda}_{3 \rightarrow L})^2 (\underline{\lambda}_{H})^{2}.
\end{equation}

Then using \eqref{eq:lower:bound} we can analyze the GD iteration. We can show $\nabla f_{\rm SL}$ is Lipschitz continuous by bounding the largest singular values of weight matrices in each layer~\cite{nguyen2020global}. To be specific, there is a positive constant $Q_0$ such that
\begin{equation*}
    \left\|\nabla f_{\rm SL}\left(\bTheta^{m+1}\right)-\nabla f_{\rm SL}\left(\bTheta^{m}\right)\right\|_{2} \leq Q_{0}\left\|\bTheta^{m+1}-\bTheta^{m}\right\|_{2}
\end{equation*}
If we choose step size $\eta<\frac{1}{Q_0}$, we have
\begin{equation}
\small
\begin{aligned}
&f_{\rm SL}\left(\bTheta^{m+1}\right)\\
&\leq f_{\rm SL}\left(\bTheta^{m}\right)+\left\langle\nabla f_{\rm SL}\left(\bTheta^{m}\right), \bTheta^{m+1}-\bTheta^{m}\right\rangle+\frac{Q_{0}}{2}\left\|\bTheta^{m+1}-\bTheta^{m}\right\|_{2}^{2} \\
&=f_{\rm SL}\left(\bTheta^{m}\right)-\eta\left\|\nabla f_{\rm SL}(\bTheta^{m})\right\|_{2}^{2}+\frac{Q_{0}}{2} \eta^{2}\left\|\nabla f_{\rm SL}(\bTheta^{m})\right\|_{2}^{2} \\
&\leq f_{\rm SL}\left(\bTheta^{m}\right)-\frac{1}{2} \eta\left\|\nabla f_{\rm SL}\left(\bTheta^{m}\right)\right\|_{2}^{2} \quad (\eta<\frac{1}{Q_0})\\
&= f_{\rm SL}\left(\bTheta^{m}\right)-\frac{1}{2} \eta\sum_{l=1}^L\left\|\nabla_{{l}} f_{\rm SL}\left(\bTheta^{m}\right)\right\|_{2}^{2}\\
&\leq f_{\rm SL}\left(\bTheta^{m}\right)-\frac{1}{2} \eta\left\|\nabla_{{2}} f_{\rm SL}\left(\bTheta^{m}\right)\right\|_{2}^{2}\quad (l=2)\\
&\leq  f_{\rm SL}\left(\bTheta^{m}\right)-\frac{1}{2}\eta \alpha_0\left\|\mathbf{f}_{L}^{m}-\mathbf{y}\right\|_{2}^2 \quad (\textrm{by }\eqref{eq:lower:bound})\nonumber\\
&= f_{\rm SL}\left(\bTheta^{m}\right)\left(1-\eta \alpha_{0}\right).
\end{aligned}
\end{equation}
This implies $f_{\rm SL}$ converges to zero at geometric rate. Finally, note $\{\bTheta^m\}_{m=1}^{\infty}$ is a Cauchy sequence, so its limit exists: 
\begin{equation*}
    \lim _{m  \rightarrow \infty}\bTheta^m = \bTheta^*_{\rm SL}.
\end{equation*}
Then, by continuity of $f_{\rm SL}$, we have
\begin{equation*}
    f_{\rm SL}\left(\bTheta^{*}_{\rm SL}\right)=f_{\rm SL}\left(\lim _{m \rightarrow \infty} \bTheta^{m}\right)=\lim _{m  \rightarrow \infty} f_{\rm SL}\left(\bTheta^{m}\right)=0.
\end{equation*}

\noindent\underline{\bf UL Training~\eqref{eq:UL:GD}:} For UL training with GD update \eqref{eq:UL:GD}, we first state our sketch of the proof. 
It is important to note that the first part of the proof cannot be used anymore because the objective function is no longer the squared loss, but the sum-rate \eqref{eq:unconstrained-unsupervised}. This function has a more complex structure, because it is no longer  strictly convex over the output of the neural network. Of course, it is also not possible to show that the loss converges linearly to global minimal. 


\noindent \textbf{Step 1:} 
At each iteration $m$, we will show that, there exists a constant $C<\infty$ such that the following holds:
\begin{equation*}
{\|\nabla f_{\rm UL}\left(\mathbf{\Theta}^{m+1}\right)-\nabla f_{\rm UL}\left(\mathbf{\Theta}^m\right)\|}_2\leq C \cdot {\|\mathbf{\Theta}^{m+1}-\mathbf{\Theta}^m\|}_2.
\end{equation*}
Denote $\mathbf{g}(\bTheta) =\frac{\partial f_{\rm UL}(\bTheta)}{\partial \tilde{\mathbf{q}}}$, by chain rule of derivative we have: 
\begin{align}\label{eq:term}
&\left\|\nabla f_{\rm UL}\left(\mathbf{\Theta}^{m+1}\right)-\nabla f_{\rm UL}\left(\mathbf{\Theta}^{m}\right)\right\|_{2}\nonumber\\
&=\left\|\mathbf{J}_L\left(\mathbf{\Theta}^{m+1}\right)^{T} \mathbf{g}\left(\mathbf{\Theta}^{m+1}\right)-\mathbf{J}_L\left(\mathbf{\Theta}^m\right)^{T} \mathbf{g}\left(\mathbf{\Theta}^m\right)\right\|_{2}\nonumber \\
& \leq|| \mathbf{g}\left(\mathbf{\Theta}^{m+1}\right)-\mathbf{g}\left(\mathbf{\Theta}^m\right)\|_{2}\left\|\mathbf{J}_L\left(\mathbf{\Theta}^{m+1}\right)\right\|_{2}\nonumber\\
& \quad +\left\|\mathbf{J}_L\left(\mathbf{\Theta}^{m+1}\right)-\mathbf{J}_L\left(\mathbf{\Theta}^m\right)\right\|_{2}\left\|\mathbf{g}\left(\mathbf{\Theta}^m\right)\right\|_{2}.
\end{align}
In the rest of the proof, we aim to bound each term in \eqref{eq:term}.\\
\underline{(Step 1.1)} We will show  $||\mathbf{J}_L(\mathbf{\Theta}^{m+1})||_2$ is bounded. 
{\small\begin{align*}
&\left\|\mathbf{J}_L\left(\mathbf{\Theta}^{m+1}\right)\right\|_{2} \stackrel{(i)}\leq \sum_{l=1}^{L}\left\|\frac{\partial \mathbf{f}_{L}\left(\mathbf{\Theta}^{m+1}\right)}{\partial \mathbf{w}_l}\right\|_{2}\\
& \stackrel{(ii)}= \sum_{l=1}^L \prod_{t=0}^{L-l-1}\|\mathbf{\Sigma}_L\left(\mathbf{W}_{L-t}^{T}(\bTheta^{m+1}) \otimes \mathbf{I}_{N}\right) \mathbf{\Sigma}_{L-t-1}\\
&\quad\times\left(\mathbf{I}_{n_{l}}\otimes \mathbf{F}_{l-1}(\bTheta^{m+1})\right)\|_2\\
&\stackrel{(iii)}\leq \sum_{l=1}^{L} \prod_{t=l+1}^{L}\left\|\mathbf{W}_{t}\left(\bTheta^{m+1}\right)\right\|_{2}\left\|\mathbf{F}_{l-1}\left(\bTheta^{m+1}\right)\right\|_{2}\\
&\stackrel{(iv)} \leq \sum_{l=1}^{L} \prod_{t=l+1}^{L}\left\|\mathbf{W}_{t}^{m+1}\right\|_{2}\|\mathbf{F}_{l-1}^{m+1}\|_F\\
& \stackrel{(v)} \leq \|\mathbf{H}\|_F \sum_{l=1}^{L} \prod_{t=l+1}^{L}\left\|\mathbf{W}_{t}^{m+1}\right\|_{2}  \prod_{t=1}^{l-1}\left\|\mathbf{W}^{m}_{t}\right\|_{2}\\
&= \|\mathbf{H}\|_{F} \sum_{l=1}^{L} \prod_{t=1,t \neq l}^{L}\left\|\mathbf{W}_{t}^{m+1}\right\|_{2} 
\end{align*}}

\noindent{where $(i)$ is because of Cauchy-Schwards inequality; $(ii)$ comes from Lemma \ref{lemma1}; $(iii)$ follows Assumption~\ref{ass:activation} that activation function at each layer satisfies $0<a'<1$ (note that the last layer with SCReLU also satisfies this condition); $(iv)$ is because Frobenius norm is always no less than $l_2$ norm; $(v)$ follows from Lemma \ref{lemma2}.}

Suppose all the weights are bounded during training, any $l \in [L]$, $\sigma_{\max}(\mathbf{W}_l^{m+1})$ is bounded. Then, it is easy to verify that the Jacobian matrix $\mathbf{J}_L\left(\bTheta^{m+1}\right)$ is bounded given fixed $N$ samples. Thus there must exists a constant $C_1<\infty$ such that  $\|\mathbf{J}_L(\mathbf{\Theta}^m)\|_2\leq C_1$.

\noindent\underline{(Step 1.2)} We show there exists a constant $C_2<\infty$ such that 
$$\left\|\mathbf{J}_L\left(\bTheta^{m+1}\right)-\mathbf{J}_L\left(\bTheta^{m}\right)\right\|_{2}\leq C_{2}\left\|\bTheta^{m+1}-\bTheta^{m}\right\|_{2}.$$ 
By Lemma \ref{lemma2}, we have
\begin{equation*}
\begin{aligned}
&\left\|\mathbf{J}_L\left(\mathbf{\Theta}^{m+1}\right)-\mathbf{J}_L\left(\mathbf{\Theta}^{m}\right)\right\|_{2}\\
&\leq\sum_{l=1}^L \left\| \frac{\partial \mathbf{f}_L\left(\mathbf{\Theta}^{m+1}\right)}{\partial \mathbf{w}_l}-\frac{\partial \mathbf{f}_L\left(\mathbf{\Theta}^m\right)}{\partial \mathbf{w}_l}\right\|_2\\
&\leq \sqrt{L}\|\mathbf{H}\|_{F} U\left(1+L \beta\|\mathbf{H}\|_{F} U + \|\mathbf{H}\|_F U\right)\left\|\mathbf{\Theta}^{m+1}-\mathbf{\Theta}^{m}\right\|_{2}
\end{aligned}
\end{equation*}
where $U=\prod_{l=1}^{L} \max \left(1, \bar{\lambda}_{l}\right)$. The first inequality is because of Cauchy-Schwards inequality; the second inequality comes from Lemma \ref{lemma2}. 
Notice that all the weights are bounded, which implies $U$ is bounded. Thus, we can find $C_2$ such that
\begin{equation*}
   \left\|\mathbf{J}_L\left(\mathbf{\Theta}^{m+1}\right)-\mathbf{J}_L\left(\mathbf{\Theta}^{m}\right)\right\|_{2}\leq C_2 \left\|\mathbf{\Theta}^{m+1}-\mathbf{\Theta}^{m}\right\|_{2}.
\end{equation*}
\underline{(Step 1.3)} Then we show $\|\mathbf{g}\left(\boldsymbol{\Theta}^{m}\right)\|_2\leq C_3$ for some constant $C_3<\infty$. Denote the sum rate of the $n$-th sample as  
$$R^{(n)}_{\bTheta}: =R\left(\mathbf{q}\left(\boldsymbol{\Theta} ;\left|\mathbf{h}^{(n)}\right|\right);\left|\mathbf{h}^{(n)}\right|\right).$$
The vectorized gradient $g(\bTheta^m)$ can be written as
{\small{
\begin{align}
\mathbf{g}(\bTheta^m) &= -\bigg(\frac{\partial R^{(1)}_{\bTheta^m}}{\partial \mathbf{q}_1^{(1)}},\cdots,\frac{\partial R^{(N)}_{\bTheta^m}}{\partial \mathbf{q}_1^{(N)}},\cdots, \frac{\partial R^{(1)}_{\bTheta^m}}{\partial \mathbf{q}_{K}^{(1)}},\cdots,\frac{\partial R^{(N)}_{\bTheta^m}}{\partial \mathbf{q}_{K}^{(N)}}\bigg)\nonumber
\end{align}}}

\noindent Note that for all $k$ and $n$, $\mathbf{q}_k\left(\boldsymbol{\Theta} ;\left|\mathbf{h}^{(n)}\right|\right)$ is bounded over $[-\alpha,1+\alpha]$ which further implies $\frac{\partial R^{(n)}_{\bTheta^m}}{{\partial \mathbf{q}_k^{(n)}}}$ is bounded and hence 
$\|\mathbf{g}\left(\boldsymbol{\Theta}^{m}\right)\|$ can be bounded by $C_3$. 
\\
\underline{(Step 1.4)} Finally, we show there exists a constant $C_4<\infty$ such that  $$\left\|\mathbf{g}(\bTheta^{m+1})-\mathbf{g}(\bTheta^m)\right\|_2\leq C_4\|\bTheta^{m+1}-\bTheta^m\|_2.$$ 

\noindent By regrading $\mathbf{g}(\bTheta)$ as a function with respect to $\mathbf{q}(\bTheta;|\mathbf{H}|)$, it is easy to check $\mathbf{g}(\bTheta)$ has a Lipschiz constant w.r.t. $\mathbf{q}(\bTheta;|\mathbf{H}|)$. This implies that there exists $C_4^{'}$ such that $$\|\mathbf{g}\left(\bTheta^{m+1}\right)-\mathbf{g}\left(\bTheta^m\right)\|_2 \leq C_4^{'}\|\mathbf{q}(\bTheta^{m+1},|\mathbf{H}|);-\mathbf{q}(\bTheta^m;|\mathbf{H}|)\|_2.$$ 
By Lemma \ref{lemma2}, we know that the following holds:
\begin{align*}
   &\left\| \operatorname{vec}\left(\mathbf{q}({\bTheta}^{m+1};|\mathbf{H}|)\right)-\operatorname{vec}\left(\mathbf{q}(\bTheta^{m};|\mathbf{H}|)\right) \right\|_2\\ &\leq
    \left\| \mathbf{q}\left(\bTheta^{m+1};|\mathbf{H}|\right)-\mathbf{q}\left(\bTheta^m;|\mathbf{H}|\right)\right\|_F\\
    &\leq \sqrt{LNK}\|\mathbf{H}\|_{F} \frac{\prod_{l=1}^{L} \bar{\lambda}_{l}}{\min _{l \in[L]} \bar{\lambda}_{l}}\left\|\bTheta^{m+1}-\bTheta^m\right\|_{2}.
\end{align*}
Suppose all the weights are bounded during training, then for any $l \in [L]$, $\bar{\lambda}_l$ is bounded. Thus there exists a constant $C_4<\infty$ such that
\begin{align}
    &\left\|\mathbf{g}\left(\bTheta^{m+1}\right)-\mathbf{g}\left(\bTheta^{m}\right)\right\|_{2}\nonumber\\
    &\leq \frac{C_4^{'}\sqrt{LNK}}{4}\|\mathbf{H}\|_{F} \frac{\prod_{l=1}^{L} \bar{\lambda}_{l}}{\min_{l \in[L]} \bar{\lambda}_{l}}\left\|\bTheta^{m+1}-\bTheta^m\right\|_{2}\nonumber\\
    &=C_4\left\|\bTheta^{m+1}-\bTheta^m\right\|_{2}.
    \end{align}
Now we have shown that
\begin{align}
&\left\|\nabla f_{\text {UL }}\left(\bTheta^{m+1}\right)-\nabla f_{\text {UL }}\left(\bTheta^{m}\right)\right\|_{2}\nonumber\\
&\leq C_1C_4\left\|\bTheta^{m+1}-\bTheta^m\right\|_2+C_2C_3\left\|\bTheta^{m+1}-\bTheta^m\right\|_2\nonumber\\
&=(C_1C_4+C_2C_3)\left\|\bTheta^{m+1}-\bTheta^m\right\|_2\nonumber
\end{align}


\noindent{\bf Step 2:} Now that we have shown that $f_{\rm UL}$ has Lipschitz gradient, we can apply the standard descent lemma. Set $\eta<\frac{1}{C_1C_4+C_2C_3}$, we have:
\begin{align*}
f_{\rm UL}\left(\bTheta^{m+1}\right)& \leq f_{\rm UL}\left(\bTheta^m\right)+\langle \nabla f_{\rm UL}\left(\bTheta^m\right),\bTheta^{m+1}-\bTheta^m\rangle\\
&\quad\quad  + \frac{1}{2}\left(C_1C_4+C_2C_3\right)\left\|\bTheta^{m+1}-\bTheta^m\right\|_2^2 \nonumber\\
&\leq f_{\rm UL}\left(\bTheta^m\right)-\frac{1}{2}\eta \left\|\nabla f_{\rm UL}\left(\bTheta^m\right)\right\|_2^2.\nonumber
\end{align*}
Summing up from $m=1,2,\cdots,M$ and divide it by $M$, we obtain:
\begin{align*}
    f_{\rm UL}\left(\bTheta^{M}\right)-f_{\rm UL}\left(\bTheta^{0}\right)\leq -\frac{\eta}{2M} \sum_{m=1}^{M} \left\|\nabla f_{\rm UL}\left(\bTheta^m\right)\right\|_2^2.
\end{align*}
This completes proof of the conclusion for UL training. 
\hfill $\qedsymbol$
}

\newpage
\onecolumn
\section{Supplementary Material}
\subsection{Description of Stationary Points}\label{app:stationary}
Before we provide the proof of the claims, let us define the KKT conditions for SL problem~\eqref{eq:supervised}, UL problem~\eqref{eq:un-supervised}, and WSR problem~\eqref{eq:wsr}. For notation simplicity, we denote $\theta$ as any parameter in the collection of parameters.

First, let us write down the Lagrangian function of SL problem \eqref{eq:supervised}.
\begin{align}
&\; L_{\rm SL}(\mathbf{P}(\bTheta ;|\mathbf{H}|),\boldsymbol{\lambda},\boldsymbol{\mu})\nonumber\\ &=\frac{1}{2}\sum_{n=1}^{N}\sum_{k=1}^{K}\left(p_k^{(n)}(\bTheta;|\bh^{(n)}|)-\bar{p}_k^{(n)}\right)^2-\sum_{n=1}^{N}\sum_{k=1}^{K}\lambda_k^{(n)} p_k(\bTheta;|\mathbf{h}^{(n)}|)
    +\sum_{n=1}^{N}\sum_{k=1}^{K}\mu_k^{(n)}(p_k(\bTheta;|\mathbf{h}^{(n)}|)-P_{\rm max}).
\end{align}
A KKT solution $\bTheta^{*}$ should satisfy that there is a tuple $(\bTheta^{*},\bar{\boldsymbol{\lambda}},\bar{\boldsymbol{\mu}})$ such that the following relations hold:
\begin{equation}\label{KKT:SL}
\left\{\begin{array}{lr}
\; \frac{\partial L_{\rm SL}\left(\mathbf{p}\left({\bTheta^{*}};|\mathbf{H}|\right), \bar{\boldsymbol{\lambda}},\bar{\boldsymbol{\mu}}\right)}{\partial \theta} =\sum_{n=1}^{N}\sum_{k=1}^{K}\left(p_k(\bTheta^{*};|\bh^{(n)}|)-\bar{p}_k^{(n)}\right)\cdot \frac{\partial p_k(\bTheta^{*};|\bh^{(n)}|)}{\partial\theta}  -\sum_{n=1}^{n}\sum_{k=1}^{K} \bar{\lambda}_{k}^{(n)} \cdot \frac{\partial p_k(\bTheta^{*};|\bh^{(n)}|)}{\partial \theta}&\\
\quad \quad \quad  \quad \quad \quad \quad \quad \quad\quad +\sum_{n=1}^{N}\sum_{k=1}^{K} \tilde{\mu}_{k}^{(n)} \cdot \frac{\partial p_k(\bTheta^{*};|\bh^{(n)}|)}{\partial \theta}=0 & \\
\;{\bf 0}   \leq \mathbf{P}(\bTheta^{*};|\mathbf{H}|)\leq {\bf P_{\rm max}}\\
\;\bar{\lambda}_{k}^{(n)} \geq 0,\ \forall n,k\\
\;\bar{\mu}_{k}^{(n)}  \geq 0,\ \forall n,k\\
\;\bar{\lambda}_{k}^{(n)} \cdot p_k(\bTheta^{*};|\bh^{(n)}|) =0,\ \forall n,k\\
\;\bar{\mu}_{k}^{(n)} \cdot \left(p_k(\bTheta^{*};|\bh^{(n)}|)-P_{\rm max}\right)=0,\ \forall n,k \end{array}\right..
\end{equation}

For UL problem \eqref{eq:un-supervised}, we write down the Lagrangian function first:
\begin{align}\label{eq:lag:UL}
    &\; L_{\rm UL}(\mathbf{P}(\bTheta ;|\mathbf{H}|),\boldsymbol{\lambda},\boldsymbol{\mu})\nonumber\\ &=\sum_{n=1}^{N}-R(\mathbf{p}(\bTheta ;|\mathbf{h}^{(n)}|),|\mathbf{h}^{(n)}|)-\sum_{n=1}^{N}\sum_{k=1}^{K}\lambda_k^{(n)} p_k(\bTheta;|\mathbf{h}^{(n)}|)
    +\sum_{n=1}^{N}\sum_{k=1}^{K}\mu_k^{(n)}(p_k(\bTheta;|\mathbf{h}^{(n)}|)-P_{\rm max}).
\end{align}
A stationary solution $\tilde{\bTheta}$ should satisfy that there is a tuple $(\tilde{\bTheta},\tilde{\boldsymbol{\lambda}},\tilde{\boldsymbol{\mu}})$ such that the following set of relations hold: 
\begin{equation}\label{KKT:UL}
\left\{\begin{array}{lr}
\; \frac{\partial L_{\rm UL}\left(\mathbf{p}\left(\tilde{\bTheta};|\mathbf{H}|\right), \tilde{\boldsymbol{\lambda}},\tilde{\boldsymbol{\mu}}\right)}{\partial \theta} =-\sum_{n=1}^{N}\sum_{k=1}^{K}\frac{R\left(\mathbf{p}\left(\tilde{\bTheta} ;|\mathbf{h}^{(n)}|\right);|\mathbf{h}^{(n)}|\right)}{\partial p_k^{(n)}}
\cdot \frac{\partial p_k(\tilde{\bTheta};|\bh^{(n)}|)}{\partial \theta}  -\sum_{n=1}^{N}\sum_{k=1}^{K} \tilde{\lambda}_{k}^{(n)} \cdot \frac{\partial p_k(\tilde{\bTheta};|\bh^{(n)}|)}{\partial \theta}&\\
\quad \quad \quad  \quad \quad \quad \quad \quad \quad\quad +\sum_{n=1}^{N}\sum_{k=1}^{K} \tilde{\mu}_{k}^{(n)} \cdot \frac{\partial p_k(\tilde{\bTheta};|\bh^{(n)}|)}{\partial \theta}=0 & \\
\;{\bf 0}   \leq \mathbf{P}(\tilde{\bTheta};\|\mathbf{H}\|)\leq {\bf P_{\rm max}}\\
\;\tilde{\lambda}_{k}^{(n)} \geq 0,\ \forall n,k\\
\;\tilde{\mu}_{k}^{(n)}  \geq 0,\ \forall n,k\\
\;\tilde{\lambda}_{k}^{(n)} \cdot p_k(\tilde{\bTheta};|\bh^{(n)}|) =0,\ \forall n,k\\
\;\tilde{\mu}_{k}^{(n)} \cdot \left(p_k(\tilde{\bTheta};|\bh^{(n)}|)-P_{\rm max}\right)=0,\ \forall n,k \end{array}\right..
\end{equation}
Similarly, we can define the KKT solution of WSR problem \eqref{eq:wsr}. Write down the Lagrangian function as:
\begin{align}\label{eq:lag:sumwsr}
    L_{\rm WSR}^{(n)}(\mathbf{p}^{(n)}, \boldsymbol{\lambda}, \boldsymbol{\mu})&= -R(\mathbf{p}^{(n)},|\mathbf{h}^{(n)}|)-\sum_{k=1}^{K} \lambda_{k}^{(n)} p_{k}^{(n)}
    +\sum_{k=1}^{K} \mu_{k}^{(n)}(p_{k}^{(n)}-P_{\rm max}).
\end{align}
A stationary solution $\bar{\bp}^{(n)}$ of \eqref{eq:wsr} for data $\bh^{(n)}$ is the point that satisfies the following conditions: there exists a tuple $(\bar{\bp}^{(n)},\bar{\boldsymbol{\lambda}}^{(n)},\bar{\boldsymbol{\mu}}^{(n)})$ such that the following holds true:
\begin{equation}\label{KKT:WSR}
\left\{\begin{array}{lr}
\;  \frac{\partial L^{(n)}_{\rm WSR}\left(\bar{\mathbf{p}}^{(n)}, \bar{\boldsymbol{\lambda}}, \bar{\boldsymbol{\mu}}\right)}{\partial p_k^{(n)}}= -\frac{ \partial R(\bar{\mathbf p}^{(n)},|\mathbf{h}^{(n)}|)}{\partial p_k^{(n)}} -  \bar{\lambda}_{k}^{(n)}+ 
\bar{\mu}_{k}^{(n)}=0\\
\;0\leq \bar{p}_k^{(n)}\leq P_{\rm max},\\
\;\bar{\lambda}_{k}^{(n)}\geq 0,\ \forall k\\
\;\bar{\mu}_{k}^{(n)}\geq 0,\ \forall k\\
\;\bar{\lambda}_{k}^{(n)} \bar{p}_k^{(n)}=0,\ \forall k\\
\;\bar{\mu}_{k}^{(n)} (\bar{p}_k^{(n)}-P_{\rm max})=0,\ \forall k
\end{array}\right.
\end{equation}

\subsection{Proof of Claim \ref{claim:ssl}}
{ \begin{proof}[Proof]
The main idea of the proof is as follows. First, we characterize properties of the solution set $\mathcal{S}$ (defined in~\eqref{soluset:sl}) of SL problem~\eqref{eq:supervised}. Second, we show that each $\bTheta\in \mathcal{S}$ also belongs to the solution set $\mathcal{L}$ (defined in~\eqref{soluset:ssl}) of SSL problem~\eqref{eq:semi}. Third, we show that each $\bTheta\in \mathcal{L}$ is also contained in the solution set $\mathcal{U}$ (defined in~\eqref{soluset:ssl}) of UL problem~\eqref{eq:un-supervised}.

First, we discuss about the solution set $\mathcal{S}$. By zero loss condition, for any $\bTheta\in\mathcal{S}$, the following holds true:
\begin{align}\label{eq:zero:loss}
   p_k(\bTheta;|\mathbf{h}^{(n)}|)  =\bar{p}_k^{(n)}, \quad \forall~n \in \mathcal{N}, \; \forall~ k\in [K].
\end{align}
Since for $n \in \mathcal{N}$, $\bar{\bp}^{(n)}$ is a stationary solution for WSR problem \eqref{eq:wsr}, so for each $n \in \mathcal{N}$, there exists a tuple $(\bar{\bp}^{(n)},\bar{\boldsymbol{\lambda}}^{(n)},\bar{\boldsymbol{\mu}}^{(n)})$ such that the KKT condition in~\eqref{KKT:WSR} holds true. This together with~\eqref{eq:zero:loss} implies the following holds for all $k\in[K],n\in\mathcal{N}$:
\begin{equation}\label{SLWSR}
\left\{\begin{array}{lr}
\;  -\frac{ \partial R(\mathbf{p}(\bTheta ;|\mathbf{h}^{(n)}|),|\mathbf{h}^{(n)}|)}{\partial p_k^{(n)}} - \bar{\lambda}_{k}^{(n)}+ 
\bar{\mu}_{k}^{(n)}=0,\\
\;0\leq p_k(\bTheta ;|\mathbf{h}^{(n)}|)\leq P_{\rm max},\\
\;\bar{\lambda}_{k}^{(n)}\geq 0,\\
\;\bar{\mu}_{k}^{(n)}\geq 0,\\
\;\bar{\lambda}_{k}^{(n)} p_k(\bTheta ;|\mathbf{h}^{(n)}|)=0,\\
\;\bar{\mu}_{k}^{(n)} (p_k(\bTheta ;|\mathbf{h}^{(n)}|)-P_{\rm max})=0.
\end{array}\right.
\end{equation}

Next, we verify that $\bTheta\in\mathcal{S}$ belongs to the solution set $\mathcal{L}$. We aim to check this by showing $\bTheta$ satisfies the KKT condition for the SSL problem~\eqref{eq:semi}. 
Towards this end, let us write down the Lagrangian for SSL problem~\eqref{eq:semi} as:
\begin{align}\label{eq:lagrange}
    \; L_{\rm SSL}(\mathbf{P}(\bTheta ;|\mathbf{H}|),\boldsymbol{\lambda},\boldsymbol{\mu})&=\sum_{n \in \mathcal{N}}-R(\mathbf{p}(\bTheta ;|\mathbf{h}^{(n)}|),|\mathbf{h}^{(n)}|)-\sum_{n \in \mathcal{N}}\sum_{k=1}^{K}\lambda_k^{(n)} p_k(\bTheta;|\mathbf{h}^{(n)}|)+\sum_{n \in \mathcal{N}}\mu_k^{(n)}(P_{\rm max}-p_k(\bTheta;|\mathbf{h}^{(n)}|))\nonumber\\
    &\quad+\sum_{m \in \mathcal{M}}\sum_{k=1}^{K}(p_k(\bTheta;|\bh^{(m)}|)-\bar{p}_k^{(m)})^2.
\end{align}
The KKT condition for the SSL problem \eqref{eq:semi} is that, there is a tuple $(\tilde{\bTheta},\tilde{\boldsymbol{\lambda}},\tilde{\boldsymbol{\mu}})$ 
such that the following set of relations holds: 
\begin{equation}\label{KKT:semi}
\left\{\begin{array}{l}
 \frac{\partial L_{\rm SSL}\left(\mathbf{P}\left(\tilde{\bTheta};|\mathbf{H}|\right), \tilde{\boldsymbol{\lambda}},\tilde{\boldsymbol{\mu}}\right)}{\partial \theta} =-\sum_{n \in \mathcal{N}}\sum_{k=1}^{K}\frac{R\left(\mathbf{p}\left(\tilde{\bTheta} ;|\mathbf{h}^{(n)}|\right);|\mathbf{h}^{(n)}|\right)}{\partial p_k^{(n)}} \cdot \frac{\partial p_k(\tilde{\bTheta};|\bh^{(n)}|)}{\partial \theta} -\sum_{n \in \mathcal{N}}\sum_{k=1}^{K} \tilde{\lambda}_{k}^{(n)} \cdot \frac{\partial p_k(\tilde{\bTheta};|\bh^{(n)}|)}{\partial \theta},\\
\quad +\sum_{n \in \mathcal{N}}\sum_{k=1}^{K} \tilde{\mu}_{k}^{(n)} \cdot \frac{\partial p_k(\tilde{\bTheta};|\bh^{(m)}|)}{\partial \theta}+2\sum_{m \in \mathcal{M}}\sum_{k=1}^{K}(p_k(\bTheta;|\bh^{(m)}|)-\bar{p}_k^{(m)})\cdot \frac{\partial p_k(\tilde{\bTheta};|\bh^{(m)}|)}{\partial \theta}=0,\\
\;{\bf 0}   \leq \mathbf{P}(\tilde{\bTheta};|\mathbf{H}|)\leq {\bf P_{\rm max}},\\
\;\tilde{\lambda}_{k}^{(n)} \geq 0,\ \forall k\in[K], n\in\mathcal{N},\\
\;\tilde{\mu}_{k}^{(n)}  \geq 0,\ \forall k\in[K], n\in\mathcal{N},\\
\;\tilde{\lambda}_{k}^{(n)} \cdot p_k(\tilde{\bTheta};|\bh^{(n)}|) =0,\ \forall k\in[K], n\in\mathcal{N},\\
\;\tilde{\mu}_{k}^{(n)} \cdot \left(p_k(\tilde{\bTheta};|\bh^{(n)}|)-P_{\rm max}\right)=0,\ \forall k\in[K], n\in\mathcal{N},\\
\; p_k(\tilde{\bTheta};|\mathbf{h}^{(m)}|)  =\bar{p}_k^{(m)}, \ \forall~ k\in [K], m \in \mathcal{M}, .\end{array}\right.
\end{equation}
To show that $\bTheta\in\mathcal{S}$ (together with some multipliers) will satisfy \eqref{KKT:semi}, we will utilize the zero-loss property in~\eqref{eq:zero:loss} and stationary condition in~\eqref{SLWSR}. 

Now we argue that $\bTheta\in\mathcal{S}$ (together with  multipliers in~\eqref{SLWSR}) satisfies the KKT condition in \eqref{KKT:semi}. The second to last equation are easy to verify. To show the first inequality, we have
\begin{align}
   \frac{\partial L_{\rm SSL}\left(\mathbf{P}\left(\bTheta;|\mathbf{H}|\right), \bar{\boldsymbol{\lambda}},\bar{\boldsymbol{\mu}}\right)}{\partial \theta} &=-\sum_{n \in \mathcal{N}}\sum_{k=1}^{K}\frac{R\left(\mathbf{p}\left(\bTheta ;|\mathbf{h}^{(n)}|\right);|\mathbf{h}^{(n)}|\right)}{\partial p_k^{(n)}} \cdot \frac{\partial p_k(\bTheta;|\bh^{(n)}|)}{\partial \theta}\nonumber \\
&\quad-\sum_{n \in \mathcal{N}}\sum_{k=1}^{K} \bar{\lambda}_{k}^{(n)} \cdot \frac{\partial p_k(\bTheta;|\bh^{(n)}|)}{\partial \theta} +\sum_{n \in \mathcal{N}}\sum_{k=1}^{K} \bar{\mu}_{k}^{(n)} \cdot \frac{\partial p_k(\bTheta;|\bh^{(n)}|)}{\partial \theta}\nonumber \\
& \quad+2\sum_{m \in \mathcal{M}}\sum_{k=1}^{K}(p_k(\bTheta;|\bh^{(m)}|)-\bar{p}_k^{(m)})\cdot \frac{\partial p_k(\bTheta;|\bh^{(n)}|)}{\partial \theta}\nonumber\\
&=\sum_{n \in \mathcal{N}}\sum_{k=1}^{K}\left(-\frac{R\left(\mathbf{p}\left(\bTheta;|\mathbf{h}^{(n)}|\right);|\mathbf{h}^{(n)}|\right)}{\partial p_k^{(n)}}- \bar{\lambda}_k^{(n)}+ \bar{\mu}_k^{(n)}\right)\times \frac{\partial p_k(\bTheta;|\bh^{(m)}|)}{\partial \theta}+0\nonumber\\
   &=0.
\end{align}
The second equality comes from the zero-loss condition in~\eqref{eq:zero:loss} and the last equation follows from the stationary condition in \eqref{SLWSR}. Hence, we have $\bTheta\in\mathcal{L}$.

Now we argue the final part of the proof. We aim to show that for any $\tilde{\bTheta}\in\mathcal{L}$, we also have $\tilde{\bTheta}\in\mathcal{U}$. To show this, we need to verify that the point satisfies the KKT condition of \eqref{eq:un-supervised}, that is, there is a tuple $(\hat{\bTheta},\hat{\boldsymbol{\lambda}},\hat{\boldsymbol{\mu}})$ satisfying the following:
\begin{equation}\label{re:KKT:UL}
\left\{\begin{array}{lr}
\; \frac{\partial L_{\rm UL}\left(\mathbf{p}\left(\hat{\bTheta};|\mathbf{h}^{(n)}|\right), \hat{\boldsymbol{\lambda}},\hat{\boldsymbol{\mu}}\right)}{\partial \theta} =\sum_{n \in \mathcal{N}}\sum_{k=1}^{K}-\frac{R\left(\mathbf{p}\left(\hat{\bTheta} ;|\mathbf{h}^{(n)}|\right);|\mathbf{h}^{(n)}|\right)}{\partial p_k^{(n)}}\cdot \frac{\partial p_k(\hat{\bTheta};|\bh^{(n)}|)}{\partial \theta}  -\sum_{n \in \mathcal{N}}\sum_{k=1}^{K} \hat{\lambda}_{k}^{(n)} \cdot \frac{\partial p_k(\hat{\bTheta};|\bh^{(n)}|)}{\partial \theta}&\\
\quad \quad  \quad  \quad  \quad  \quad  +\sum_{n \in \mathcal{N}}\sum_{k=1}^{K} \hat{\mu}_{k}^{(n)} \cdot \frac{\partial p_k(\hat{\bTheta};|\bh^{(n)}|)}{\partial \theta}=0, & \\
\;{\bf 0}   \leq \mathbf{P}(\hat{\bTheta};|\mathbf{H}|)\leq {\bf P_{\rm max}},\\
\;\hat{\lambda}_{k}^{(n)} \geq 0,\ \forall k\in[K], n\in\mathcal{N},\\
\;\hat{\mu}_{k}^{(n)}  \geq 0,\ \forall k\in[K], n\in\mathcal{N},\\
\;\hat{\lambda}_{k}^{(n)} \cdot p_k(\hat{\bTheta};|\bh^{(n)}|) =0,\ \forall k\in[K], n\in\mathcal{N},\\
\;\hat{\mu}_{k}^{(n)} \cdot \left(p_k(\hat{\bTheta};|\bh^{(n)}|)-P_{\rm max}\right)=0,\ \forall k\in[K], n\in\mathcal{N}. \end{array}\right.
\end{equation}
Suppose that the tuple $(\tilde{\bTheta},\tilde{\boldsymbol{\lambda}},\tilde{\boldsymbol{\mu}})$ satisfies the KKT condition in \eqref{KKT:semi}.
By the zero-loss regularization condition for samples $m\in\mathcal{M}$, we have
\begin{equation}
   p_k(\tilde{\bTheta};|\mathbf{h}^{(m)}|)  =\bar{p}_k^{(m)}, \quad \forall~m \in \mathcal{M}, \; \forall~ k\in [K].\nonumber 
\end{equation}
Then it is easy to verify that the tuple $(\tilde{\bTheta},\tilde{\boldsymbol{\lambda}},\tilde{\boldsymbol{\mu}})$ satisfies the second to last equation in \eqref{re:KKT:UL}. Now we only need to show the first inequality.
\begin{align*}
    &\frac{\partial L_{\mathrm{UL}}\left(\mathbf{P}\left(\tilde{\bTheta} ;\left|\mathbf{H}\right|\right), \tilde{\boldsymbol{\lambda}}, \tilde{\boldsymbol{\mu}}\right)}{\partial \theta} \\
    &= \sum_{n \in \mathcal{N}}\sum_{k=1}^{K}\bigg(-\frac{ \partial R\left(\bp(\tilde{\bTheta};|\bh^{(n)}|),|\mathbf{h}^{(n)}|\right)}{\partial p_k^{(n)}} - \tilde{\lambda}_{k}^{(n)}+  \tilde{\mu}_{k}^{(n)}\bigg)\times\frac{\partial p_{k}\left(\tilde{\bTheta} ; |\mathbf{h}^{(n)}|\right)}{\partial \theta}\\
    &=\frac{\partial L_{\rm SSL}\left(\mathbf{P}\left({\tilde{\bTheta}};|\mathbf{H}|\right), \tilde{\boldsymbol{\lambda}},\tilde{\boldsymbol{\mu}}\right)}{\partial \theta}-2\sum_{m \in \mathcal{M}}\sum_{k=1}^{K}(p_k(\tilde{\bTheta};|\bh^{(m)}|)-\bar{p}_k^{(m)})\times\frac{\partial p_{k}\left(\tilde{\bTheta}; |\mathbf{h}^{(m)}|\right)}{\partial \theta}\\
   &=\frac{\partial L_{\rm SSL}\left(\mathbf{P}\left(\tilde{\bTheta};|\mathbf{H}|\right), \tilde{\boldsymbol{\lambda}},\tilde{\boldsymbol{\mu}}\right)}{\partial \theta}\\
   &=0.
\end{align*}
where the last equation comes from the zero-loss regularization condition. Hence, we have $\tilde{\bTheta}\in\mathcal{U}$.

\end{proof}}

\end{document}